\documentclass[american,aps,prx,reprint,floatfix,nofootinbib,superscriptaddress,notitlepage]{revtex4-2}
\usepackage[unicode=true,pdfusetitle, bookmarks=true,bookmarksnumbered=false,bookmarksopen=false, breaklinks=false,pdfborder={0 0 0},backref=false,colorlinks=false]{hyperref}
\hypersetup{colorlinks,linkcolor=myurlcolor,citecolor=myurlcolor,urlcolor=myurlcolor}
\usepackage{graphics,epstopdf,graphicx,amsthm,amsmath,amssymb,
braket,colortbl,color,bm,framed,mathrsfs}

\usepackage[T1]{fontenc}
\usepackage[up]{subfigure}
\usepackage{tikz}
\usepackage{algorithm}
\usepackage{algorithmic}
\usepackage{enumerate}
\usepackage{multirow}

\usepackage{tcolorbox}

\definecolor{myurlcolor}{rgb}{0,0,0.9}

\newcommand{\proj}[1]{| #1\rangle\!\langle #1 |}

\newcommand{\inner}[2]{\langle #1 , #2\rangle}

\DeclareMathOperator{\trace}{Tr}
\newcommand{\Ptr}[2]{\trace_{#1}\Pa{#2}}
\newcommand{\Tr}[1]{\Ptr{}{#1}}

\newcommand{\Pa}[1]{\left[#1\right]}

\theoremstyle{plain}
\newtheorem{thm}{Theorem}

\newtheorem{prop}[thm]{Proposition}

\newtheorem{Exam}[thm]{Example}
  {\begin{Exam}\upshape}{\end{Exam}}

\newcommand*{\myproofname}{Proof}

\def\ot{\otimes}
\def\complex{\mathbb{C}}
\def\real{\mathbb{R}}

\usepackage{ulem,xpatch}
\xpatchcmd{\sout}
  {\bgroup}
  {\bgroup}
  {}{}
\newcommand{\be}{\begin{equation}}
\newcommand{\ee}{\end{equation}}
\newcommand{\beq}{\begin{eqnarray}}
\newcommand{\eeq}{\end{eqnarray}}

\DeclareMathAlphabet{\mathcal}{OMS}{cmsy}{m}{n}

\makeatother

\begin{document}
\title{Extremality of stabilizer states}
\author{Kaifeng Bu}
\email{kfbu@fas.harvard.edu}
\affiliation{\it Department of Physics, Harvard University, Cambridge, MA 02138, USA
}

\begin{abstract}

We investigate the extremality of stabilizer states to reveal their exceptional role in the space of all 
$n$-qubit/qudit states. 
We establish uncertainty principles for the characteristic function and the Wigner function of states, respectively. We find that only stabilizer states achieve saturation
in these principles.
Furthermore, we prove a general theorem that stabilizer states are extremal for convex information measures invariant under local unitaries. 
We explore this extremality in the context of various quantum information and correlation measures, including entanglement entropy, conditional entropy and 
other entanglement measures.
Additionally, leveraging the recent discovery that stabilizer states are the limit states under quantum convolution,  we 
establish the monotonicity of the entanglement entropy and conditional entropy under quantum convolution.
These results highlight the remarkable information-theoretic properties of stabilizer states. Their extremality provides valuable insights into their ability to capture information content and correlations, paving the way for further exploration of their potential in quantum information processing.

\end{abstract}

\maketitle

\newpage

\section{Introduction}
Stabilizer states stand as a cornerstone, offering a rich theoretical framework and practical utility across a variety of quantum computing models and error correction schemes. Originating from the study of quantum error-correcting codes \cite{Gottesman97}, stabilizer states are defined through the stabilizer formalism, which utilizes a commuting group of Pauli operators to uniquely determine a quantum state. 
This formalism not only simplifies the representation of quantum states but also facilitates efficient simulation of quantum circuits under certain conditions.

Stabilizer states encompass a broad class of quantum states, including the Bell state and GHZ state, which play a
fundamental role in quantum entanglement theory. 
 The nice mathematical structure of stabilizer states allows for their elegant description and manipulation. This makes them invaluable in the development of quantum error correction, where they form the basis of stabilizer codes—a class of codes capable of protecting quantum information from errors. Shor's  9-qubit code~\cite{ShorPRA95} and Kitaev's toric code~\cite{Kitaev_toric} are two well-known examples of stabilizer codes.
Recently, stabilizer codes have been successfully realized in experiments~\cite{Google2023suppressing,Bluvstein2024logical}, marking a significant milestone towards achieving fault-tolerant quantum computation.
 
Furthermore, the stabilizer formalism underpins the Gottesman-Knill theorem~\cite{gottesman1998heisenberg}, which posits that any quantum circuit composed solely of
stabilizer input states, Clifford gates (which map stabilizer states to stabilizer states) and Pauli measurement,  can be simulated in polynomial time on a classical computer. This theorem not only delineates the boundary between classical and quantum computational power but also emphasizes the importance of non-Clifford gates in achieving quantum computational advantage. Later, the
extension of the Gottesman-Knill theorem beyond stabilizer circuits was studied to further understand the boundary between classical and quantum computation~\cite{BravyiPRL16,BravyiPRX16,bravyi2019simulation,BeverlandQST20,SeddonPRXQ21, bu2022classical,gao2018efficient,Bu19,UmeshSTOC23,koh2015further}.

In this work, we focus on the extremality of stabilizer states.
The concept of extremality has been well-studied within the context of Gaussian states in bosonic systems~\cite{Holevo99,HolevoMutual99,CerfPRL04,Eisert2005gaussian,WolfPRL06}.
A prominent example is the maximum entropy principle: for states with a fixed covariance matrix, Gaussian states achieve the maximum quantum entropy~\cite{Holevo99}. 
Similar extremality properties have been shown for other information and correlation measures, including the mutual information~\cite{HolevoMutual99}, conditional entropy~\cite{CerfPRL04,Eisert2005gaussian} and entanglement measure~\cite{WolfPRL06}.
 However, little is known about the extremality of stabilizer states. 

We investigate the extremality of stabilizer states by starting from some uncertainty principles.
We establish uncertainty principles 
specifically tailored to stabilizer states. 
These principles connect the max-entropy of a state with the characteristic function and the Wigner function of the 
states.  Interestingly, the derived inequalities are saturated only by stabilizer states.

Moreover, we present a general theorem that guarantees that stabilizer states are extremal for any convex information measure invariant under local unitary operations. 
This extremality can be applied to various information and correlation measures, including entanglement entropy, conditional entropy, and convex entanglement measures. We explore how these measures behave for stabilizer states compared to their non-stabilizer counterparts, revealing their ability to capture quantum information and correlations.

Furthermore, as the stabilizer states are the limit states in the quantum central limit theorem \cite{BGJ23a,BGJ23b, BGJ24a}, 
we study the behavior of the information measures under quantum convolution. We show the monotonicity of the entanglement entropy and conditional entropy 
 under quantum convolution.

\section{Preliminary}
We focus on an $n$-qudit system with Hilbert space $\mathcal{H}^{\ot n}$. Here $\mathcal{H} \simeq \complex^d$ is  $d$-dimensional, and $d$ could be any integer number greater than or equal to $2$. Let $D(\mathcal{H}^{\ot n})$ denote the set of  all quantum states on $\mathcal{H}^{\ot n}$.
We consider the orthonormal, computational basis in $\mathcal{H}$ denoted by $\set{\ket{k}}_{k\in \mathbb{Z}_d}$. The Pauli $X$ and $Z$ operators
are 
$X: |k\rangle\mapsto |k+1 \rangle,  Z: |k\rangle \mapsto\omega^k_d\,|k\rangle,\;\forall k\in \mathbb{Z}_d\;. 
$
Here  $\mathbb{Z}_{d}$ is the cyclic group over $d$,  and $\omega_d=\exp(2\pi i /d)$ is a $d$-th root of unity. 
If  $d$ is an odd prime number,  the local Pauli operators
are defined as 
$
w(p,q)=\omega^{-2^{-1}pq}_d\, Z^pX^q
$. Here $2^{-1}$ denotes the inverse $\frac{d+1}{2}$ of 2 in $\mathbb{Z}_d$.
If $d=2$, the Pauli operators are defined as 
$
w(p,q)=i^{-pq}Z^pX^q
$. Pauli operators for general local dimension $d$ can be found in~\cite{Jaffe17}.  In the $n$-qudit system, the $n$-qudit Pauli  operators are defined as
$
w(\vec p, \vec q)
=w(p_1, q_1)\ot...\ot w(p_n, q_n)$,
with $\vec p=(p_1, p_2,..., p_n)\in \mathbb{Z}^n_d$, $\vec q=(q_1,..., q_n)\in \mathbb{Z}^n_d $.

Denote $V^n:=\mathbb{Z}^n_d\times \mathbb{Z}^n_d$; this represents the phase space for $n$-qudit systems~\cite{Gross06}. The set of Pauli operators forms an orthonormal basis 
 with respect to the inner product 
$\inner{A}{B}=\frac{1}{d^n}\Tr{A^\dag B}$. 
The characteristic function $\Xi_{\rho}:V^{n}\to\complex$ of a quantum state $\rho$ is
\begin{eqnarray*}
\Xi_{\rho}(\vec{p},\vec q):=\Tr{\rho w(\vec{p},\vec q)^\dag}.
\end{eqnarray*}
The support of $\Xi_{\rho}$ is given by
$\text{Supp}(\Xi_{\rho})=\set{\vec x:\Xi_{\rho}(\vec x)\neq 0}$.

Any quantum state $\rho$ can be  expressed as a linear combination of the Pauli operators 
$
\rho=\frac{1}{d^n}
\sum_{(\vec{p},\vec q)\in V^n}
\Xi_{\rho}(\vec{p},\vec q)w(\vec{p},\vec q)\;.
$
The Pauli operators can serve as a Fourier basis, which has
 found extensive uses in various  applications, including quantum Boolean functions~\cite{montanaro2010quantum}, quantum circuit complexity~\cite{Bucomplexity22}, quantum scrambling~\cite{GBJPNAS23}, the generalization capacity 
 of quantum machine learning~\cite{BuPRA19_stat}, and
  quantum state 
 tomography~\cite{Bunpj22}.

In this work, we focus on stabilizer states.
A pure stabilizer state $\ket{\psi}$ is identified as  the common eigenstate of a commuting  subgroup of the 
Pauli operators with  $n$ generators $ \set{g_i}_{i\in [n]}$, i.e., $g_i\ket{\psi}=\ket{\psi}$ for each $i$. 
The corresponding density matrix can be expressed 
as $\proj{\psi}=
\Pi^n_{i=1}\mathbb{E}_{k_i\in \mathbb{Z}_d}g^{k_i}_i$, where the expectation $\mathbb{E}_{k_i\in \mathbb{Z}_d}g^{k_i}_i:=\frac{1}{d}\sum_{k_i\in \mathbb{Z}_d}g^{k_i}_i$.
In general, a mixed state $\rho$ is called a stabilizer state if 
there exists some commuting  subgroup of the 
Pauli operators  with $r\leq n$ generators $ \set{g_i}_{i\in [r]}$ such that $\rho=\frac{1}{d^{n-r}}
\Pi^r_{i=1}\mathbb{E}_{k_i\in \mathbb{Z}_d}g^{k_i}_i$. We take STAB to denote the set of all stabilizer states.
Note that some literature considers states that can be written as any convex combination of stabilizers. 
However, such states fall outside the scope of this work as they generally lose the stabilizer formalism.

 Now, given a quantum state $\rho$,  let us consider the set of Pauli operators $\set{w(\vec x): |\Xi_{\rho}(\vec x)|=1,\vec x\in V^n}$, denoted as $G_{\rho}$. It was shown that the 
 $G_{\rho}$ is a commuting subgroup of Pauli operators (See lemma 12 in \cite{BGJ23b}). 
 Hence, we call $G_{\rho}$ the stabilizer group of $\rho$.
 The stabilizer group is closely related to the concept of mean states~\cite{BGJ23a,BGJ23b}. For an $n$-qudit state $\rho$, its mean state $\mathcal{M}(\rho)$ is defined by the characteristic function as follows:
\begin{align}\label{0109shi6}
\Xi_{\mathcal{M}(\rho)}(\vec x) :=
\left\{
\begin{aligned}
&\Xi_\rho ( \vec x) , && |\Xi_\rho ( \vec x)|=1,\\
& 0 , && |\Xi_\rho (  \vec x)|<1.
\end{aligned}
\right.
\end{align}
It was shown that  $\mathcal{M}(\rho)$  is a stabilizer state \cite{BGJ23a,BGJ23b}.
Hence $\mathcal{M}(\rho)$ is the stabilizer state with the same stabilizer group as $\rho$.
Furthermore, a state $\rho$ belongs to the set of stabilizer states (STAB)  iff $\rho=\mathcal{M}(\rho)$.

\section{Main results}
We start by introducing uncertainty principles for quantum states using their characteristic and Wigner functions. For the sake of clarity, all proofs  are provided in the Appendix.
 
 {\bf \textit{Extremality of stabilizer states in uncertainty principles.---}}Given a quantum state $\rho$, 
 we can consider its rank, denoted by $\text{Rank}(\rho)$, as the number of nonzero eigenvalues. The logarithm of the rank is referred to as the max-entropy~\cite{Renner05}, denoted by  $S_{\max}(\rho)$. Another important concept is the Pauli rank, denoted by $\chi_P(\rho)$ \cite{Bu19}.
 This refers to the nonzero coefficients in the Pauli decomposition, i.e., the size of the support of its characteristic function.
 We now explore the relationship between max-entropy $S_{\max}(\rho)$ and the Pauli rank 
$\chi_P(\rho)$.

\begin{thm}[\bf Uncertainty principle for Pauli rank ]\label{thm:UP}
    Given an $n$-qudit state $\rho$, we have
    \begin{eqnarray}
S_{\max}(\rho)+\log \chi_P(\rho)
        \geq n\log d,
    \end{eqnarray}
    and the equality holds iff $\rho$ is a stabilizer state, i.e.,  $\rho=\mathcal{M}(\rho)$.
\end{thm}

This uncertainty principle leverages the fact that the characteristic function can be viewed as a set of quantum Fourier coefficients. It reveals that stabilizer states are the only quantum states that achieve the equality in this specific uncertainty relation.
It is worth noting that in the Heisenberg uncertainty principle~\cite{Heisenberg1927anschaulichen}, only Gaussian wave packets can achieve equality (see \cite{ColesPRMP17}  and references therein). This finding suggests that stabilizer states might be considered the ``discrete quantum Gaussians'' introduced in~\cite{BGJ23a}.

Furthermore, for a pure state $\rho=\proj{\psi}$, the rank $ \text{Rank}(\proj{\psi})$ is always equal to $1$. Consequently, the Pauli rank,
$\chi_P(\proj{\psi})$ must be greater than or equal to $d^n$.
Equality holds only if the pure state $\ket{\psi} $ is a stabilizer state. 
Combined with the invariance under Clifford unitaries, this uncertainty 
principle guarantees that the Pauli rank $\chi_P$ can be used to quantify non-stabilizerness of quantum states, as demonstrated in \cite{Bu19}.

The uncertainty principle in Theorem \ref{thm:UP} applies to all the local dimensions $d$. 
However, for odd prime dimensions, we can leverage the (discrete) Wigner function~\cite{Gross06}.
Within the phase space, we define phase space point operators denoted by $T(\vec x)$ with $\vec x\in V^n$. 
These operators are defined as  $T(\vec x)=w(\vec x)T(\vec 0) w(\vec x)^\dag$, and $T(\vec 0)$ is given by $T(\vec 0) = \frac{1}{d^n}\sum_{\vec u\in V^n}w(\vec{u})$.  
The set $\set{T(\vec x)}_{\vec x\in V^n}$ forms an orthonormal basis  with respect to the inner
product $\inner{A}{B}=\frac{1}{d^n}\Tr{A^\dag B}$.
This allows us to express any quantum state $\rho$ as a linear combination of these operators:
$\rho=\sum_{\vec x}W_{\rho}(\vec x)T(\vec x)$, where $W_{\rho}(\vec x)=\inner{T(\vec x)}{\rho}$ is the Wigner function. 
 Notably, the Wigner function is real-valued due to the Hermitian nature of the operators $T(\vec x)$.

Similar to the Pauli rank, 
 we can define the Wigner rank $\chi_W(\rho)$ as the size of the support of the  Wigner function $W_{\rho}$.
Then, we can establish another uncertainty principle between the max-entropy $S_{\max}$ and 
the Wigner rank.

\begin{thm}[\bf Uncertainty principle for Wigner rank ]
Given an $n$-qudit state $\rho$ with odd prime $d$, we have 
    \begin{eqnarray}
      S_{\max}(\rho)+\log \chi_W(\rho)
        \geq n\log d,
    \end{eqnarray}
    and the equality holds iff $\rho$ is a pure stabilizer state.
\end{thm}
Different from Theorem \ref{thm:UP},
the uncertainty principle  for Wigner rank 
is only saturated by pure stabilizer states.
Note that, for pure state $\proj{\psi}$, $\chi_W(\proj{\psi})\geq d^n$, with equality iff 
$\psi$ is pure stabilizer states. Therefore, akin to the Pauli rank \cite{Bu19}, the
Wigner rank can also serve as a tool to quantify non-stabilizerness.

The characteristic function and the Wigner function are related through the symplectic Fourier transform~\cite{Gross06}. This connection allows us to establish the following uncertainty principle using both the Pauli rank and the Wigner rank, where the
stabilizer states achieve saturation. For a concise overview, Table~\ref{tab:sum_A} summarizes these three uncertainty principles.

\begin{prop}\label{prop:xw}
    Given an $n$-qudit state $\rho$ with odd prime $d$, we have 
    \begin{eqnarray}
        \log \chi_P(\rho)+\log \chi_W(\rho)
        \geq 2n\log d,
    \end{eqnarray}
    and the equality holds iff $\rho$ is a stabilizer state, i.e.,  $\rho=\mathcal{M}(\rho)$.
\end{prop}

\begin{table}[!htbp]
\centering
\begin{tabular}{ |c|c|c|c| } 
\hline
\multirow{3}{*}{Measure}& \\
&Uncertainty principle\\
&\\
\hline
\multirow{4}{*}{ $S_{\max}$ vs $\Xi_{\rho}$}&   \\
& $S_{\max}(\rho)+\log \chi_P(\rho)\geq n\log d$  \\
&``='' iff $\rho$ is a stabilizer state     \\
&\\
\hline
\multirow{4}{*}{$S_{\max}$ vs  $W_{\rho}$}&  \\
& $S_{\max}(\rho)+\log \chi_W(\rho)\geq n\log d$  \\
& ``='' iff  $\rho$ is a pure stabilizer state   \\
&\\
\hline
\multirow{4}{*}{$\Xi_{\rho}$ vs $W_{\rho}$}&   \\
&  $\log \chi_P(\rho)+\log \chi_W(\rho)\geq 2n\log d$    \\
& ``=''  iff  $\rho$ is a stabilizer state \\
&\\
\hline
\end{tabular}
\caption{\label{tab:sum_A}Summary of the uncertainty principles.}

\end{table}

 {\bf \textit{Extremality of stabilizer states in information and correlation measures.---}}
Now, let us explore the extremality of stabilizer states in the quantum information and correlation measures.
We begin with a general result:

\begin{thm}[\bf General result]\label{thm:extre}
   Let $F:D(\otimes^m_i\mathcal{H}_i)\to\real$  be a convex function, 
    where $\otimes^m_i\mathcal{H}_i$ is an $m$-partite system and each subsystem $\mathcal{H}_i$ consists of $n_i$ qudits. If $F$ is
    invariant under local unitary,
    we have 
    \begin{eqnarray}
        F(\rho)\geq F(\mathcal{M}(\rho)),
    \end{eqnarray}
    where $\mathcal{M}(\rho)$ is the stabilizer  state with the same stabilizer group as $\rho$.
\end{thm}
It is important to note that the condition of having the same stabilizer group for $n$-qudit/qubit states plays a similar role as having a fixed covariance matrix for Bosonic states. Both conditions identify a specific class of states within a larger space.

 Theorem \ref{thm:extre} provides a straightforward way
to rederive 
the maximum entropy principle for $n$-qudit/qubit states: $S(\rho)\leq S(\mathcal{M}(\rho))$ (See \cite{BGJ23a,BGJ23b}) by taking $F$ to be minus the quantum 
entropy $S(\rho):=-\Tr{\rho\log \rho}$. 
Furthermore,  it has been shown that $\mathcal{M}(\rho)$  is the closest stabilizer state with respect to the relative entropy $
D(\rho||\sigma):=\Tr{\rho(\log\rho-\log\sigma)}$, i.e., $\min\limits_{\sigma\in\text{STAB}}D(\rho||\sigma)=S(\mathcal{M}(\rho))-S(\rho)$~\cite{BGJ23a,BGJ23b}.

We will now explore the application of Theorem~\ref{thm:extre} to other information and correlation measures. A summary of these applications will be provided in Table~\ref{tab:sum_B}.

{\bf Example: Entanglement entropy--}
One application of Theorem~\ref{thm:extre} is to entanglement entropy. Entanglement entropy, denoted by 
$S(A)_{\rho}$, is defined as the quantum entropy of $\rho_A$, where $\rho_A=\Ptr{A^c}{\rho}$ and 
$A^c$ is the complement of the subsystem $A$. It is a popular measure for quantifying entanglement, particularly in pure states. Many other entanglement measures reduce to entanglement entropy for pure states. 
However, for mixed states, entanglement entropy is not a perfect measure of entanglement. 
Here, Theorem~\ref{thm:extre} comes into play. We can verify that the negative of entanglement entropy, $-S(A)_{\rho}$, 
is both convex and invariant under local unitary operations. Consequently, Theorem~\ref{thm:extre} guarantees the following inequality:
\begin{align}
  S(A)_{\rho}\leq S(A)_{\mathcal{M}(\rho)}.  
\end{align}

{\bf Example: Conditional entropy--}
Another application of Theorem~\ref{thm:extre} is to conditional entropy, denoted by $S(A|B)_{\rho}$. 
Conditional entropy measures the amount of uncertainty about subsystem $A$ after knowing the state of subsystem $B$. 
The negative conditional entropy $-S(A|B)_{\rho}$ is a valuable tool for quantifying quantum correlations\cite{Cerf97,Konig2009operational}. These correlations have diverse applications in areas like quantum communication~\cite{Horodecki2005partial}, quantum thermodynamics~\cite{Rio2011thermodynamic}, and quantum cryptography~\cite{Brown2021computing}.

We can verify that the negative conditional entropy, $-S(A|B)_{\rho}$, satisfies the conditions of Theorem~\ref{thm:extre} (convexity and local unitary invariance). Therefore, the theorem guarantees the following inequality:
\begin{eqnarray}
    S(A|B)_{\rho}\leq S(A|B)_{\mathcal{M}(\rho)}.
\end{eqnarray}
That is,  a state's conditional entropy is always less than or equal to the conditional entropy of its mean state.

The concept can be extended to the more general R\'enyi conditional entropy, defined as
$S_{\alpha}(A|B)_{\rho}=-\inf_{\sigma_B}D_{\alpha}(\rho_{AB}||I_A\ot \sigma_B)$ \cite{Tomamichel2015quantum1}, where the R\'enyi relative entropy $D_{\alpha}(\rho||\sigma):=\frac{1}{\alpha-1}\log\Tr{\left(\sigma^{\frac{1-\alpha}{2\alpha}}\rho\sigma^{\frac{1-\alpha}{2\alpha}}\right)^{\alpha}}$, for $\alpha\in [0, +\infty]$.
 One specific member of this family is the min-entropy $S_{\min}(A|B)_{\rho}:=S_{\infty}(A|B)_{\rho}$, where  $\exp(-S_{\min}(A|B)_{\rho})$ captures the maximum overlap achievable with the Bell state (a maximally entangled state) through local operations on subsystem $B$~\cite{Konig2009operational}. 
 Since 
$S_{\alpha}(A|B)$ is concave for $\alpha\geq 1/2$,  Theorem \ref{thm:extre} again applies:
\begin{eqnarray}
    S_{\alpha}(A|B)_{\rho}\leq S_{\alpha}(A|B)_{\mathcal{M}(\rho)}.
\end{eqnarray}

{\bf Example: Entanglement measure--}
Theorem~\ref{thm:extre} can also be applied to various entanglement measures. Since bipartite and multipartite entanglement measures typically satisfy local unitary invariance, the theorem can be leveraged under certain conditions.
Hence, if the entanglement measure $E$ is convex,  Theorem \ref{thm:extre} guarantees the following inequality:
\begin{eqnarray}
    E(\rho)\geq E(\mathcal{M}(\rho)).
\end{eqnarray}
That is, the entanglement measure of a state is always greater than or equal to the entanglement measure of its mean state.
For any entanglement measure defined by the convex roof\footnote{The convex roof starts with  a measure
$E$ on pure states, and then extends it to mixed 
ones by $E(\rho):=\inf\sum_ip_iE(\proj{\psi_i})$, where the infimum is taken over all the pure state decomposition $\rho=\sum_ip_i\proj{\psi_i}$.}, it must satisfy the convexity, such as the entanglement of formation~\cite{BennettPRA96}. In addition, 
many other entanglement measures, such as negativity~\cite{Zyczkowski98,Vidal02}, and squashed entanglement~\cite{Christandl2004squashed}, are also convex, and thus Theorem~\ref{thm:extre} still applies.

\begin{table}[!htbp]
\centering
\begin{tabular}{ |c|c|c|c| } 
\hline
 \multirow{3}{*}{Information/Correlation Measure} &  \\
 & Extremality \\
 & \\
\hline
\multirow{3}{*}{Quantum entropy}&   \\
& $S(\rho)\leq S(\mathcal{M}(\rho))$  \\
& \cite{BGJ23a,BGJ23b}\\
&\\
\hline
\multirow{3}{*}{Entanglement entropy}&  \\
& $S(A)_{\rho}\leq S(A)_{\mathcal{M}(\rho)}$   \\
&     \\
\hline
\multirow{3}{*}{Conditional  entropy}&   \\
&  $S_{\alpha}(A|B)_{\rho}\leq S_{\alpha}(A|B)_{\mathcal{M}(\rho)}$  \\
&    \\
\hline
\multirow{3}{*}{Convex entanglement measure}&   \\
&  $E(\rho)\geq E(\mathcal{M}(\rho))$  \\
&    \\
\hline
\end{tabular}
\caption{\label{tab:sum_B}Summary of the extremality of stabilizer states.}

\end{table}

 {\bf \textit{Monotonocity of the information and correlation measures under quantum convolution.---}}
 Recently, a new quantum convolution has been introduced to study stabilizer states and channels~\cite{BGJ23a,BGJ23b,BGJ23c,BGJ24a,BJ24a,BJW24a}.
  This convolution applies to both qubit and qudit systems. We focus on the $n$-qudit system in this work, but the results can be extended to $n$-qubit systems using the qubit-convolution defined in~\cite{BGJ23c}.

 The quantum convolution  for qudit systems is defined as follows \cite{BGJ23a,BGJ23b}:
 Given $s,t \in \mathbb{Z}_d$ which satisfy $s^2+t^2\equiv 1 ~({\rm mod}~d)$, the  unitary  operator  $U_{s,t}$  acting on  a $2n$-qudit system $\mathcal{H}_A\ot \mathcal{H}_B $ is
$U_{s,t}:\ket{\vec i}\ot\ket{\vec j}\to \ket{s\vec i+t\vec j  \mod d}\ot \ket{- t\vec i+s\vec j\mod d} $,
where   both $\mathcal{H}_A$ and $\mathcal{H}_B$ are $n$-qudit systems.     
The convolution of two $n$-qudit states $\rho$ and $\sigma$ is 
\begin{align}\label{eq:conv_B}
\rho \boxtimes_{s,t} \sigma = \Ptr{B}{ U_{s,t} (\rho \otimes \sigma) U^\dag_{s,t}}.
\end{align}
We focus on non-trivial parameters $s,t$, i.e., neither is $0$ or $1$. For simplicity,
we denote  $\boxtimes_{s,t}$  by  $\boxtimes$. 

Additionally, for a quantum state $\rho$, $\boxtimes_{L}\rho$ denotes the $L$-th repetition of the quantum convolution, 
defined as 
$\boxtimes_{L}\rho=\boxtimes_{L-1}\rho \boxtimes \rho$ where $\boxtimes_0\rho:=\rho$.
A key result related to the quantum convolution is the quantum central limit theorem. This theorem states that the repeated convolution of a state,
 $\boxtimes_L\rho$, converges to a stabilizer state $\mathcal{M}(\rho)$.
 
Theorem~\ref{thm:extre} shows that under certain conditions, $F(\rho)\geq F(\mathcal{M}(\rho))$  (or $-F(\rho)\leq -F(\mathcal{M}(\rho))$).
This leads us to investigate the behavior of information and correlation measures under convolution. It has already been established that quantum entropy increases monotonically under quantum convolution~\cite{BGJ23a,BGJ23b}. In this work, we explore two other information and correlation measures under convolution.
 A summary of the monotonicity of these information measures is provided in Table~\ref{tab:sum_C}.

\begin{prop}[\bf Monotonicity of entanglement entropy]\label{prop:EEM}
 The entanglement entropy is monotonically increasing under quantum convolution, i.e., 
 \begin{eqnarray}
     S(A)_{\boxtimes_L\rho}
     \leq   S(A)_{\boxtimes_{L+1}\rho},
 \end{eqnarray}
 for any $L\geq 0$, where $S(A)_{\boxtimes_L\rho}=S((\boxtimes_L\rho)_A)$ is the entanglement entropy of subsystem $A$ after applying the $L$-th iteration of the quantum convolution to the state $\rho$.
\end{prop}

\begin{prop}[\bf Monotonicity of conditional entropy]\label{prop:CEM}
 The conditional entropy $S_{\alpha}(A|B)$ is monotonically increasing under quantum  convolution for any $\alpha\geq 1/2$, i.e., 
 \begin{eqnarray}
     S_{\alpha}(A|B)_{\boxtimes_L\rho}
     \leq   S_{\alpha}(A|B)_{\boxtimes_{L+1}\rho},
 \end{eqnarray}
 for any $L\geq 0$.
\end{prop}

Quantum convolution has also been used to define a new coarse-graining procedure called the Convolution Group (CG)~\cite{BJW24a}. This concept shares similarities with the Renormalization Group (RG) framework. In the CG framework, the mean state $\mathcal{M}(\rho)$ acts as the fixed point. Notably, the CG has been applied to classify different quantum states~\cite{BJW24a}.

In the context of the RG, establishing monotonicity under a coarse-graining procedure is crucial, particularly for understanding irreversibility. Examples include the c-theorem for the 2D RG flow~\cite{Zamolodchikov1986irreversibility} and its extensions to higher dimensions~\cite{Cardy1988there,Osborn1989derivation,Komargodski2011renormalization,NishiokaRMP18}. Therefore, investigating monotonicity under the CG is also of interest in this context.

Our propositions, Propositions~\ref{prop:EEM} and~\ref{prop:CEM}, can be used to establish the monotonicity of entanglement entropy and conditional entropy under the CG. Furthermore, since $\exp(-S_{\alpha}(A|B))$ can be used to quantify quantum correlation, our propositions also imply that the 
quantum correlation is monotonically decreasing under the CG. 

\begin{table}[!htbp]
\centering
\begin{tabular}{ |c|c|c|c| } 
\hline
\multirow{3}{*}{Information measure} & \\
& Monotonicity\\
&\\
\hline
\multirow{4}{*}{Quantum entropy}&   \\
& $S(\boxtimes_L\rho)\leq S(\boxtimes_{L+1}\rho)$   \\
&\cite{BGJ23a,BGJ23b}\\
&\\
\hline
\multirow{3}{*}{Entanglement entropy}&  \\
& $S(A)_{\boxtimes_L\rho}\leq S(A)_{\boxtimes_{L+1}\rho}$   \\
&     \\
\hline
\multirow{3}{*}{Conditional  entropy}&   \\
&  $S_{\alpha}(A|B)_{\boxtimes_L\rho}\leq S_{\alpha}(A|B)_{\boxtimes_{L+1}\rho}$  \\
&\\
\hline
\end{tabular}
\caption{\label{tab:sum_C}Summary of the monotonicity of the information measures under quantum convolution $\boxtimes$.}

\end{table}

\section{Conclusion and Discussion}

The concept of extremality for stabilizer states, along with the associated quantum central limit theorem, sheds light on why these states can be considered discrete quantum Gaussians. This makes them particularly important in quantum information theory and quantum statistical mechanics.

Our work further demonstrates that deviations from the extremality of stabilizer states can be used to quantify the degree of non-stabilizerness in a state. These novel uncertainty principles, combined with the established monotonicity of correlations under convolution, offer valuable new insights into the nature of stabilizer states.

This work establishes the uncertainty principles for max-entropy, Pauli rank, and Wigner rank. Extending these results to other entropy measures and investigating the effect of memory on the uncertainty principle is an interesting avenue for future exploration (cf.  \cite{Berta2010uncertainty}).

Furthermore, the Choi–Jamiołkowski isomorphism ~\cite{Choi75,Jamio72} suggests the possibility of analyzing the extremality of Clifford unitaries (or stabilizer channels) using a similar approach. Additionally, generalizing the observed monotonicity of entanglement entropy under convolution to other entanglement measures, such as the entanglement of formation, is another promising direction for future research.

\section*{Acknowledgements}
The author thanks Diandian Wang,  Yuanjie Ren, and Roy Garcia for the discussion.
 KB is supported by the ARO Grant W911NF-19-1-0302 and the ARO
MURI Grant W911NF-20-1-0082.

\bibliography{reference}{}

\begin{thebibliography}{59}%
\makeatletter
\providecommand \@ifxundefined [1]{%
 \@ifx{#1\undefined}
}%
\providecommand \@ifnum [1]{%
 \ifnum #1\expandafter \@firstoftwo
 \else \expandafter \@secondoftwo
 \fi
}%
\providecommand \@ifx [1]{%
 \ifx #1\expandafter \@firstoftwo
 \else \expandafter \@secondoftwo
 \fi
}%
\providecommand \natexlab [1]{#1}%
\providecommand \enquote  [1]{``#1''}%
\providecommand \bibnamefont  [1]{#1}%
\providecommand \bibfnamefont [1]{#1}%
\providecommand \citenamefont [1]{#1}%
\providecommand \href@noop [0]{\@secondoftwo}%
\providecommand \href [0]{\begingroup \@sanitize@url \@href}%
\providecommand \@href[1]{\@@startlink{#1}\@@href}%
\providecommand \@@href[1]{\endgroup#1\@@endlink}%
\providecommand \@sanitize@url [0]{\catcode `\\12\catcode `\$12\catcode
  `\&12\catcode `\#12\catcode `\^12\catcode `\_12\catcode `\%12\relax}%
\providecommand \@@startlink[1]{}%
\providecommand \@@endlink[0]{}%
\providecommand \url  [0]{\begingroup\@sanitize@url \@url }%
\providecommand \@url [1]{\endgroup\@href {#1}{\urlprefix }}%
\providecommand \urlprefix  [0]{URL }%
\providecommand \Eprint [0]{\href }%
\providecommand \doibase [0]{https://doi.org/}%
\providecommand \selectlanguage [0]{\@gobble}%
\providecommand \bibinfo  [0]{\@secondoftwo}%
\providecommand \bibfield  [0]{\@secondoftwo}%
\providecommand \translation [1]{[#1]}%
\providecommand \BibitemOpen [0]{}%
\providecommand \bibitemStop [0]{}%
\providecommand \bibitemNoStop [0]{.\EOS\space}%
\providecommand \EOS [0]{\spacefactor3000\relax}%
\providecommand \BibitemShut  [1]{\csname bibitem#1\endcsname}%
\let\auto@bib@innerbib\@empty
\bibitem [{\citenamefont {Gottesman}(1997)}]{Gottesman97}%
  \BibitemOpen
  \bibfield  {author} {\bibinfo {author} {\bibfnamefont {D.}~\bibnamefont
  {Gottesman}},\ }\bibfield  {title} {\bibinfo {title} {Stabilizer codes and
  quantum error correction},\ }\href {https://arxiv.org/abs/quant-ph/9705052}
  {\bibfield  {journal} {\bibinfo  {journal} {arXiv:quant-ph/9705052}\ }
  (\bibinfo {year} {1997})}\BibitemShut {NoStop}%
\bibitem [{\citenamefont {Shor}(1995)}]{ShorPRA95}%
  \BibitemOpen
  \bibfield  {author} {\bibinfo {author} {\bibfnamefont {P.~W.}\ \bibnamefont
  {Shor}},\ }\bibfield  {title} {\bibinfo {title} {Scheme for reducing
  decoherence in quantum computer memory},\ }\href
  {https://doi.org/10.1103/PhysRevA.52.R2493} {\bibfield  {journal} {\bibinfo
  {journal} {Phys. Rev. A}\ }\textbf {\bibinfo {volume} {52}},\ \bibinfo
  {pages} {R2493} (\bibinfo {year} {1995})}\BibitemShut {NoStop}%
\bibitem [{\citenamefont {Kitaev}(2003)}]{Kitaev_toric}%
  \BibitemOpen
  \bibfield  {author} {\bibinfo {author} {\bibfnamefont {A.}~\bibnamefont
  {Kitaev}},\ }\bibfield  {title} {\bibinfo {title} {Fault-tolerant quantum
  computation by anyons},\ }\href
  {https://doi.org/10.1016/S0003-4916(02)00018-0} {\bibfield  {journal}
  {\bibinfo  {journal} {Annals of Physics}\ }\textbf {\bibinfo {volume}
  {303}},\ \bibinfo {pages} {2} (\bibinfo {year} {2003})}\BibitemShut {NoStop}%
\bibitem [{Goo(2023)}]{Google2023suppressing}%
  \BibitemOpen
  \bibfield  {title} {\bibinfo {title} {Suppressing quantum errors by scaling a
  surface code logical qubit},\ }\href
  {https://doi.org/10.1038/s41586-022-05434-1} {\bibfield  {journal} {\bibinfo
  {journal} {Nature}\ }\textbf {\bibinfo {volume} {614}},\ \bibinfo {pages}
  {676} (\bibinfo {year} {2023})}\BibitemShut {NoStop}%
\bibitem [{\citenamefont {Bluvstein}\ \emph {et~al.}(2024)\citenamefont
  {Bluvstein}, \citenamefont {Evered}, \citenamefont {Geim}, \citenamefont
  {Li}, \citenamefont {Zhou}, \citenamefont {Manovitz}, \citenamefont {Ebadi},
  \citenamefont {Cain}, \citenamefont {Kalinowski}, \citenamefont {Hangleiter}
  \emph {et~al.}}]{Bluvstein2024logical}%
  \BibitemOpen
  \bibfield  {author} {\bibinfo {author} {\bibfnamefont {D.}~\bibnamefont
  {Bluvstein}}, \bibinfo {author} {\bibfnamefont {S.~J.}\ \bibnamefont
  {Evered}}, \bibinfo {author} {\bibfnamefont {A.~A.}\ \bibnamefont {Geim}},
  \bibinfo {author} {\bibfnamefont {S.~H.}\ \bibnamefont {Li}}, \bibinfo
  {author} {\bibfnamefont {H.}~\bibnamefont {Zhou}}, \bibinfo {author}
  {\bibfnamefont {T.}~\bibnamefont {Manovitz}}, \bibinfo {author}
  {\bibfnamefont {S.}~\bibnamefont {Ebadi}}, \bibinfo {author} {\bibfnamefont
  {M.}~\bibnamefont {Cain}}, \bibinfo {author} {\bibfnamefont {M.}~\bibnamefont
  {Kalinowski}}, \bibinfo {author} {\bibfnamefont {D.}~\bibnamefont
  {Hangleiter}}, \emph {et~al.},\ }\bibfield  {title} {\bibinfo {title}
  {Logical quantum processor based on reconfigurable atom arrays},\ }\href
  {https://doi.org/10.1038/s41586-023-06927-3} {\bibfield  {journal} {\bibinfo
  {journal} {Nature}\ }\textbf {\bibinfo {volume} {626}},\ \bibinfo {pages}
  {58} (\bibinfo {year} {2024})}\BibitemShut {NoStop}%
\bibitem [{\citenamefont {Gottesman}(1998)}]{gottesman1998heisenberg}%
  \BibitemOpen
  \bibfield  {author} {\bibinfo {author} {\bibfnamefont {D.}~\bibnamefont
  {Gottesman}},\ }\bibfield  {title} {\bibinfo {title} {{The Heisenberg
  representation of quantum computers}}\ }(\bibinfo {year} {1998})\ pp.\
  \bibinfo {pages} {32--43},\ \Eprint {https://arxiv.org/abs/quant-ph/9807006}
  {arXiv:quant-ph/9807006} \BibitemShut {NoStop}%
\bibitem [{\citenamefont {Bravyi}\ and\ \citenamefont
  {Gosset}(2016)}]{BravyiPRL16}%
  \BibitemOpen
  \bibfield  {author} {\bibinfo {author} {\bibfnamefont {S.}~\bibnamefont
  {Bravyi}}\ and\ \bibinfo {author} {\bibfnamefont {D.}~\bibnamefont
  {Gosset}},\ }\bibfield  {title} {\bibinfo {title} {Improved classical
  simulation of quantum circuits dominated by {Clifford} gates},\ }\href
  {https://doi.org/10.1103/PhysRevLett.116.250501} {\bibfield  {journal}
  {\bibinfo  {journal} {Phys. Rev. Lett.}\ }\textbf {\bibinfo {volume} {116}},\
  \bibinfo {pages} {250501} (\bibinfo {year} {2016})}\BibitemShut {NoStop}%
\bibitem [{\citenamefont {Bravyi}\ \emph {et~al.}(2016)\citenamefont {Bravyi},
  \citenamefont {Smith},\ and\ \citenamefont {Smolin}}]{BravyiPRX16}%
  \BibitemOpen
  \bibfield  {author} {\bibinfo {author} {\bibfnamefont {S.}~\bibnamefont
  {Bravyi}}, \bibinfo {author} {\bibfnamefont {G.}~\bibnamefont {Smith}},\ and\
  \bibinfo {author} {\bibfnamefont {J.~A.}\ \bibnamefont {Smolin}},\ }\bibfield
   {title} {\bibinfo {title} {Trading classical and quantum computational
  resources},\ }\href {https://doi.org/10.1103/PhysRevX.6.021043} {\bibfield
  {journal} {\bibinfo  {journal} {Phys. Rev. X}\ }\textbf {\bibinfo {volume}
  {6}},\ \bibinfo {pages} {021043} (\bibinfo {year} {2016})}\BibitemShut
  {NoStop}%
\bibitem [{\citenamefont {Bravyi}\ \emph {et~al.}(2019)\citenamefont {Bravyi},
  \citenamefont {Browne}, \citenamefont {Calpin}, \citenamefont {Campbell},
  \citenamefont {Gosset},\ and\ \citenamefont {Howard}}]{bravyi2019simulation}%
  \BibitemOpen
  \bibfield  {author} {\bibinfo {author} {\bibfnamefont {S.}~\bibnamefont
  {Bravyi}}, \bibinfo {author} {\bibfnamefont {D.}~\bibnamefont {Browne}},
  \bibinfo {author} {\bibfnamefont {P.}~\bibnamefont {Calpin}}, \bibinfo
  {author} {\bibfnamefont {E.}~\bibnamefont {Campbell}}, \bibinfo {author}
  {\bibfnamefont {D.}~\bibnamefont {Gosset}},\ and\ \bibinfo {author}
  {\bibfnamefont {M.}~\bibnamefont {Howard}},\ }\bibfield  {title} {\bibinfo
  {title} {Simulation of quantum circuits by low-rank stabilizer
  decompositions},\ }\href {https://doi.org/10.22331/q-2019-09-02-181}
  {\bibfield  {journal} {\bibinfo  {journal} {{Quantum}}\ }\textbf {\bibinfo
  {volume} {3}},\ \bibinfo {pages} {181} (\bibinfo {year} {2019})}\BibitemShut
  {NoStop}%
\bibitem [{\citenamefont {Beverland}\ \emph {et~al.}(2020)\citenamefont
  {Beverland}, \citenamefont {Campbell}, \citenamefont {Howard},\ and\
  \citenamefont {Kliuchnikov}}]{BeverlandQST20}%
  \BibitemOpen
  \bibfield  {author} {\bibinfo {author} {\bibfnamefont {M.}~\bibnamefont
  {Beverland}}, \bibinfo {author} {\bibfnamefont {E.}~\bibnamefont {Campbell}},
  \bibinfo {author} {\bibfnamefont {M.}~\bibnamefont {Howard}},\ and\ \bibinfo
  {author} {\bibfnamefont {V.}~\bibnamefont {Kliuchnikov}},\ }\bibfield
  {title} {\bibinfo {title} {Lower bounds on the non-{Clifford} resources for
  quantum computations},\ }\href {https://doi.org/10.1088/2058-9565/ab8963}
  {\bibfield  {journal} {\bibinfo  {journal} {Quantum Sci. Technol.}\ }\textbf
  {\bibinfo {volume} {5}},\ \bibinfo {pages} {035009} (\bibinfo {year}
  {2020})}\BibitemShut {NoStop}%
\bibitem [{\citenamefont {Seddon}\ \emph {et~al.}(2021)\citenamefont {Seddon},
  \citenamefont {Regula}, \citenamefont {Pashayan}, \citenamefont {Ouyang},\
  and\ \citenamefont {Campbell}}]{SeddonPRXQ21}%
  \BibitemOpen
  \bibfield  {author} {\bibinfo {author} {\bibfnamefont {J.~R.}\ \bibnamefont
  {Seddon}}, \bibinfo {author} {\bibfnamefont {B.}~\bibnamefont {Regula}},
  \bibinfo {author} {\bibfnamefont {H.}~\bibnamefont {Pashayan}}, \bibinfo
  {author} {\bibfnamefont {Y.}~\bibnamefont {Ouyang}},\ and\ \bibinfo {author}
  {\bibfnamefont {E.~T.}\ \bibnamefont {Campbell}},\ }\bibfield  {title}
  {\bibinfo {title} {Quantifying quantum speedups: Improved classical
  simulation from tighter magic monotones},\ }\href
  {https://doi.org/10.1103/PRXQuantum.2.010345} {\bibfield  {journal} {\bibinfo
   {journal} {PRX Quantum}\ }\textbf {\bibinfo {volume} {2}},\ \bibinfo {pages}
  {010345} (\bibinfo {year} {2021})}\BibitemShut {NoStop}%
\bibitem [{\citenamefont {Bu}\ and\ \citenamefont
  {Koh}(2022)}]{bu2022classical}%
  \BibitemOpen
  \bibfield  {author} {\bibinfo {author} {\bibfnamefont {K.}~\bibnamefont
  {Bu}}\ and\ \bibinfo {author} {\bibfnamefont {D.~E.}\ \bibnamefont {Koh}},\
  }\bibfield  {title} {\bibinfo {title} {Classical simulation of quantum
  circuits by half {G}auss sums},\ }\href
  {https://doi.org/10.1007/s00220-022-04320-1} {\bibfield  {journal} {\bibinfo
  {journal} {Commun. Math. Phys.}\ }\textbf {\bibinfo {volume} {390}},\
  \bibinfo {pages} {471} (\bibinfo {year} {2022})}\BibitemShut {NoStop}%
\bibitem [{\citenamefont {Gao}\ and\ \citenamefont
  {Duan}(2018)}]{gao2018efficient}%
  \BibitemOpen
  \bibfield  {author} {\bibinfo {author} {\bibfnamefont {X.}~\bibnamefont
  {Gao}}\ and\ \bibinfo {author} {\bibfnamefont {L.}~\bibnamefont {Duan}},\
  }\bibfield  {title} {\bibinfo {title} {Efficient classical simulation of
  noisy quantum computation},\ }\href {https://arxiv.org/abs/1810.03176}
  {\bibfield  {journal} {\bibinfo  {journal} {arXiv:1810.03176}\ } (\bibinfo
  {year} {2018})}\BibitemShut {NoStop}%
\bibitem [{\citenamefont {Bu}\ and\ \citenamefont {Koh}(2019)}]{Bu19}%
  \BibitemOpen
  \bibfield  {author} {\bibinfo {author} {\bibfnamefont {K.}~\bibnamefont
  {Bu}}\ and\ \bibinfo {author} {\bibfnamefont {D.~E.}\ \bibnamefont {Koh}},\
  }\bibfield  {title} {\bibinfo {title} {Efficient classical simulation of
  {C}lifford circuits with nonstabilizer input states},\ }\href
  {https://doi.org/10.1103/PhysRevLett.123.170502} {\bibfield  {journal}
  {\bibinfo  {journal} {Phys. Rev. Lett.}\ }\textbf {\bibinfo {volume} {123}},\
  \bibinfo {pages} {170502} (\bibinfo {year} {2019})}\BibitemShut {NoStop}%
\bibitem [{\citenamefont {Aharonov}\ \emph {et~al.}(2023)\citenamefont
  {Aharonov}, \citenamefont {Gao}, \citenamefont {Landau}, \citenamefont
  {Liu},\ and\ \citenamefont {Vazirani}}]{UmeshSTOC23}%
  \BibitemOpen
  \bibfield  {author} {\bibinfo {author} {\bibfnamefont {D.}~\bibnamefont
  {Aharonov}}, \bibinfo {author} {\bibfnamefont {X.}~\bibnamefont {Gao}},
  \bibinfo {author} {\bibfnamefont {Z.}~\bibnamefont {Landau}}, \bibinfo
  {author} {\bibfnamefont {Y.}~\bibnamefont {Liu}},\ and\ \bibinfo {author}
  {\bibfnamefont {U.}~\bibnamefont {Vazirani}},\ }\bibfield  {title} {\bibinfo
  {title} {A polynomial-time classical algorithm for noisy random circuit
  sampling},\ }\href {https://doi.org/10.1145/3564246.3585234} {\bibfield
  {journal} {\bibinfo  {journal} {STOC}\ ,\ \bibinfo {pages} {945–957}}
  (\bibinfo {year} {2023})}\BibitemShut {NoStop}%
\bibitem [{\citenamefont {Koh}(2017)}]{koh2015further}%
  \BibitemOpen
  \bibfield  {author} {\bibinfo {author} {\bibfnamefont {D.~E.}\ \bibnamefont
  {Koh}},\ }\bibfield  {title} {\bibinfo {title} {Further extensions of
  {C}lifford circuits and their classical simulation complexities},\ }\href
  {https://doi.org/10.26421/QIC17.3-4} {\bibfield  {journal} {\bibinfo
  {journal} {Quantum Information \& Computation}\ }\textbf {\bibinfo {volume}
  {17}},\ \bibinfo {pages} {0262} (\bibinfo {year} {2017})}\BibitemShut
  {NoStop}%
\bibitem [{\citenamefont {Holevo}\ \emph {et~al.}(1999)\citenamefont {Holevo},
  \citenamefont {Sohma},\ and\ \citenamefont {Hirota}}]{Holevo99}%
  \BibitemOpen
  \bibfield  {author} {\bibinfo {author} {\bibfnamefont {A.~S.}\ \bibnamefont
  {Holevo}}, \bibinfo {author} {\bibfnamefont {M.}~\bibnamefont {Sohma}},\ and\
  \bibinfo {author} {\bibfnamefont {O.}~\bibnamefont {Hirota}},\ }\bibfield
  {title} {\bibinfo {title} {Capacity of quantum {G}aussian channels},\ }\href
  {https://doi.org/10.1103/PhysRevA.59.1820} {\bibfield  {journal} {\bibinfo
  {journal} {Phys. Rev. A}\ }\textbf {\bibinfo {volume} {59}},\ \bibinfo
  {pages} {1820} (\bibinfo {year} {1999})}\BibitemShut {NoStop}%
\bibitem [{\citenamefont {Holevo}\ and\ \citenamefont
  {Werner}(2001)}]{HolevoMutual99}%
  \BibitemOpen
  \bibfield  {author} {\bibinfo {author} {\bibfnamefont {A.~S.}\ \bibnamefont
  {Holevo}}\ and\ \bibinfo {author} {\bibfnamefont {R.~F.}\ \bibnamefont
  {Werner}},\ }\bibfield  {title} {\bibinfo {title} {Evaluating capacities of
  bosonic {G}aussian channels},\ }\href
  {https://doi.org/10.1103/PhysRevA.63.032312} {\bibfield  {journal} {\bibinfo
  {journal} {Phys. Rev. A}\ }\textbf {\bibinfo {volume} {63}},\ \bibinfo
  {pages} {032312} (\bibinfo {year} {2001})}\BibitemShut {NoStop}%
\bibitem [{\citenamefont {Grosshans}\ and\ \citenamefont
  {Cerf}(2004)}]{CerfPRL04}%
  \BibitemOpen
  \bibfield  {author} {\bibinfo {author} {\bibfnamefont {F.}~\bibnamefont
  {Grosshans}}\ and\ \bibinfo {author} {\bibfnamefont {N.~J.}\ \bibnamefont
  {Cerf}},\ }\bibfield  {title} {\bibinfo {title} {Continuous-variable quantum
  cryptography is secure against non-{G}aussian attacks},\ }\href
  {https://doi.org/10.1103/PhysRevLett.92.047905} {\bibfield  {journal}
  {\bibinfo  {journal} {Phys. Rev. Lett.}\ }\textbf {\bibinfo {volume} {92}},\
  \bibinfo {pages} {047905} (\bibinfo {year} {2004})}\BibitemShut {NoStop}%
\bibitem [{\citenamefont {Eisert}\ and\ \citenamefont
  {Wolf}(2005)}]{Eisert2005gaussian}%
  \BibitemOpen
  \bibfield  {author} {\bibinfo {author} {\bibfnamefont {J.}~\bibnamefont
  {Eisert}}\ and\ \bibinfo {author} {\bibfnamefont {M.~M.}\ \bibnamefont
  {Wolf}},\ }\bibfield  {title} {\bibinfo {title} {Gaussian quantum channels},\
  }\href@noop {} {\bibfield  {journal} {\bibinfo  {journal} {arXiv preprint
  quant-ph/0505151}\ } (\bibinfo {year} {2005})}\BibitemShut {NoStop}%
\bibitem [{\citenamefont {Wolf}\ \emph {et~al.}(2006)\citenamefont {Wolf},
  \citenamefont {Giedke},\ and\ \citenamefont {Cirac}}]{WolfPRL06}%
  \BibitemOpen
  \bibfield  {author} {\bibinfo {author} {\bibfnamefont {M.~M.}\ \bibnamefont
  {Wolf}}, \bibinfo {author} {\bibfnamefont {G.}~\bibnamefont {Giedke}},\ and\
  \bibinfo {author} {\bibfnamefont {J.~I.}\ \bibnamefont {Cirac}},\ }\bibfield
  {title} {\bibinfo {title} {Extremality of {G}aussian quantum states},\ }\href
  {https://doi.org/10.1103/PhysRevLett.96.080502} {\bibfield  {journal}
  {\bibinfo  {journal} {Phys. Rev. Lett.}\ }\textbf {\bibinfo {volume} {96}},\
  \bibinfo {pages} {080502} (\bibinfo {year} {2006})}\BibitemShut {NoStop}%
\bibitem [{\citenamefont {Bu}\ \emph {et~al.}(2023{\natexlab{a}})\citenamefont
  {Bu}, \citenamefont {Gu},\ and\ \citenamefont {Jaffe}}]{BGJ23a}%
  \BibitemOpen
  \bibfield  {author} {\bibinfo {author} {\bibfnamefont {K.}~\bibnamefont
  {Bu}}, \bibinfo {author} {\bibfnamefont {W.}~\bibnamefont {Gu}},\ and\
  \bibinfo {author} {\bibfnamefont {A.}~\bibnamefont {Jaffe}},\ }\bibfield
  {title} {\bibinfo {title} {Quantum entropy and central limit theorem},\
  }\href {https://doi.org/10.1073/pnas.2304589120} {\bibfield  {journal}
  {\bibinfo  {journal} {Proceedings of the National Academy of Sciences}\
  }\textbf {\bibinfo {volume} {120}},\ \bibinfo {pages} {e2304589120} (\bibinfo
  {year} {2023}{\natexlab{a}})}\BibitemShut {NoStop}%
\bibitem [{\citenamefont {Bu}\ \emph {et~al.}(2023{\natexlab{b}})\citenamefont
  {Bu}, \citenamefont {Gu},\ and\ \citenamefont {Jaffe}}]{BGJ23b}%
  \BibitemOpen
  \bibfield  {author} {\bibinfo {author} {\bibfnamefont {K.}~\bibnamefont
  {Bu}}, \bibinfo {author} {\bibfnamefont {W.}~\bibnamefont {Gu}},\ and\
  \bibinfo {author} {\bibfnamefont {A.}~\bibnamefont {Jaffe}},\ }\bibfield
  {title} {\bibinfo {title} {Discrete quantum {G}aussians and central limit
  theorem},\ }\href {https://arxiv.org/abs/2302.08423} {\bibfield  {journal}
  {\bibinfo  {journal} {arXiv:2302.08423}\ } (\bibinfo {year}
  {2023}{\natexlab{b}})}\BibitemShut {NoStop}%
\bibitem [{\citenamefont {Bu}\ \emph {et~al.}(2024{\natexlab{a}})\citenamefont
  {Bu}, \citenamefont {Gu},\ and\ \citenamefont {Jaffe}}]{BGJ24a}%
  \BibitemOpen
  \bibfield  {author} {\bibinfo {author} {\bibfnamefont {K.}~\bibnamefont
  {Bu}}, \bibinfo {author} {\bibfnamefont {W.}~\bibnamefont {Gu}},\ and\
  \bibinfo {author} {\bibfnamefont {A.}~\bibnamefont {Jaffe}},\ }\bibfield
  {title} {\bibinfo {title} {Entropic quantum central limit theorem and quantum
  inverse sumset theorem},\ }\href {https://arxiv.org/abs/2401.14385}
  {\bibfield  {journal} {\bibinfo  {journal} {arXiv:2401.14385}\ } (\bibinfo
  {year} {2024}{\natexlab{a}})}\BibitemShut {NoStop}%
\bibitem [{\citenamefont {Liu}\ \emph {et~al.}(2017)\citenamefont {Liu},
  \citenamefont {Wozniakowski},\ and\ \citenamefont {Jaffe}}]{Jaffe17}%
  \BibitemOpen
  \bibfield  {author} {\bibinfo {author} {\bibfnamefont {Z.}~\bibnamefont
  {Liu}}, \bibinfo {author} {\bibfnamefont {A.}~\bibnamefont {Wozniakowski}},\
  and\ \bibinfo {author} {\bibfnamefont {A.~M.}\ \bibnamefont {Jaffe}},\
  }\bibfield  {title} {\bibinfo {title} {Quon 3d language for quantum
  information},\ }\href {https://doi.org/10.1073/pnas.1621345114} {\bibfield
  {journal} {\bibinfo  {journal} {PNAS}\ }\textbf {\bibinfo {volume} {114}},\
  \bibinfo {pages} {2497} (\bibinfo {year} {2017})}\BibitemShut {NoStop}%
\bibitem [{\citenamefont {Gross}(2006)}]{Gross06}%
  \BibitemOpen
  \bibfield  {author} {\bibinfo {author} {\bibfnamefont {D.}~\bibnamefont
  {Gross}},\ }\bibfield  {title} {\bibinfo {title} {Hudson's theorem for
  finite-dimensional quantum systems},\ }\href
  {https://doi.org/10.1063/1.2393152} {\bibfield  {journal} {\bibinfo
  {journal} {J. Math. Phys.}\ }\textbf {\bibinfo {volume} {47}},\ \bibinfo
  {pages} {122107} (\bibinfo {year} {2006})}\BibitemShut {NoStop}%
\bibitem [{\citenamefont {Montanaro}\ and\ \citenamefont
  {Osborne}(2010)}]{montanaro2010quantum}%
  \BibitemOpen
  \bibfield  {author} {\bibinfo {author} {\bibfnamefont {A.}~\bibnamefont
  {Montanaro}}\ and\ \bibinfo {author} {\bibfnamefont {T.~J.}\ \bibnamefont
  {Osborne}},\ }\bibfield  {title} {\bibinfo {title} {Quantum {B}oolean
  functions},\ }\href@noop {} {\bibfield  {journal} {\bibinfo  {journal}
  {Chicago Journal of Theoretical Computer Science}\ }\textbf {\bibinfo
  {volume} {2010}} (\bibinfo {year} {2010})}\BibitemShut {NoStop}%
\bibitem [{\citenamefont {Bu}\ \emph {et~al.}(2022{\natexlab{a}})\citenamefont
  {Bu}, \citenamefont {Garcia}, \citenamefont {Jaffe}, \citenamefont {Koh},\
  and\ \citenamefont {Li}}]{Bucomplexity22}%
  \BibitemOpen
  \bibfield  {author} {\bibinfo {author} {\bibfnamefont {K.}~\bibnamefont
  {Bu}}, \bibinfo {author} {\bibfnamefont {R.~J.}\ \bibnamefont {Garcia}},
  \bibinfo {author} {\bibfnamefont {A.}~\bibnamefont {Jaffe}}, \bibinfo
  {author} {\bibfnamefont {D.~E.}\ \bibnamefont {Koh}},\ and\ \bibinfo {author}
  {\bibfnamefont {L.}~\bibnamefont {Li}},\ }\bibfield  {title} {\bibinfo
  {title} {Complexity of quantum circuits via sensitivity, magic, and
  coherence},\ }\href {https://doi.org/10.48550/arXiv.2204.12051} {\bibfield
  {journal} {\bibinfo  {journal} {arXiv:2204.12051}\ } (\bibinfo {year}
  {2022}{\natexlab{a}})}\BibitemShut {NoStop}%
\bibitem [{\citenamefont {Garcia}\ \emph {et~al.}(2023)\citenamefont {Garcia},
  \citenamefont {Bu},\ and\ \citenamefont {Jaffe}}]{GBJPNAS23}%
  \BibitemOpen
  \bibfield  {author} {\bibinfo {author} {\bibfnamefont {R.~J.}\ \bibnamefont
  {Garcia}}, \bibinfo {author} {\bibfnamefont {K.}~\bibnamefont {Bu}},\ and\
  \bibinfo {author} {\bibfnamefont {A.}~\bibnamefont {Jaffe}},\ }\bibfield
  {title} {\bibinfo {title} {Resource theory of quantum scrambling},\ }\href
  {https://doi.org/10.1073/pnas.2217031120} {\bibfield  {journal} {\bibinfo
  {journal} {Proceedings of the National Academy of Sciences}\ }\textbf
  {\bibinfo {volume} {120}},\ \bibinfo {pages} {e2217031120} (\bibinfo {year}
  {2023})}\BibitemShut {NoStop}%
\bibitem [{\citenamefont {Bu}\ \emph {et~al.}(2022{\natexlab{b}})\citenamefont
  {Bu}, \citenamefont {Koh}, \citenamefont {Li}, \citenamefont {Luo},\ and\
  \citenamefont {Zhang}}]{BuPRA19_stat}%
  \BibitemOpen
  \bibfield  {author} {\bibinfo {author} {\bibfnamefont {K.}~\bibnamefont
  {Bu}}, \bibinfo {author} {\bibfnamefont {D.~E.}\ \bibnamefont {Koh}},
  \bibinfo {author} {\bibfnamefont {L.}~\bibnamefont {Li}}, \bibinfo {author}
  {\bibfnamefont {Q.}~\bibnamefont {Luo}},\ and\ \bibinfo {author}
  {\bibfnamefont {Y.}~\bibnamefont {Zhang}},\ }\bibfield  {title} {\bibinfo
  {title} {Statistical complexity of quantum circuits},\ }\href
  {https://doi.org/10.1103/PhysRevA.105.062431} {\bibfield  {journal} {\bibinfo
   {journal} {Phys. Rev. A}\ }\textbf {\bibinfo {volume} {105}},\ \bibinfo
  {pages} {062431} (\bibinfo {year} {2022}{\natexlab{b}})}\BibitemShut
  {NoStop}%
\bibitem [{\citenamefont {Bu}\ \emph {et~al.}(2024{\natexlab{b}})\citenamefont
  {Bu}, \citenamefont {Koh}, \citenamefont {Garcia},\ and\ \citenamefont
  {Jaffe}}]{Bunpj22}%
  \BibitemOpen
  \bibfield  {author} {\bibinfo {author} {\bibfnamefont {K.}~\bibnamefont
  {Bu}}, \bibinfo {author} {\bibfnamefont {D.~E.}\ \bibnamefont {Koh}},
  \bibinfo {author} {\bibfnamefont {R.}~\bibnamefont {Garcia}},\ and\ \bibinfo
  {author} {\bibfnamefont {A.}~\bibnamefont {Jaffe}},\ }\bibfield  {title}
  {\bibinfo {title} {Classical shadows with pauli-invariant unitary
  ensembles},\ }\href {https://doi.org/10.1038/s41534-023-00801-w} {\bibfield
  {journal} {\bibinfo  {journal} {npj Quantum Information}\ }\textbf {\bibinfo
  {volume} {10}},\ \bibinfo {pages} {6} (\bibinfo {year}
  {2024}{\natexlab{b}})}\BibitemShut {NoStop}%
\bibitem [{\citenamefont {Renner}(2005)}]{Renner05}%
  \BibitemOpen
  \bibfield  {author} {\bibinfo {author} {\bibfnamefont {R.}~\bibnamefont
  {Renner}},\ }\bibfield  {title} {\bibinfo {title} {Security of quantum key
  distribution},\ }\href {https://arxiv.org/abs/quant-ph/0512258} {\bibfield
  {journal} {\bibinfo  {journal} {arXiv:quant-ph/0512258}\ } (\bibinfo {year}
  {2005})}\BibitemShut {NoStop}%
\bibitem [{\citenamefont {Heisenberg}(1927)}]{Heisenberg1927anschaulichen}%
  \BibitemOpen
  \bibfield  {author} {\bibinfo {author} {\bibfnamefont {W.}~\bibnamefont
  {Heisenberg}},\ }\bibfield  {title} {\bibinfo {title} {{\"U}ber den
  anschaulichen inhalt der quantentheoretischen kinematik und mechanik},\
  }\href {https://doi.org/10.1007/BF01397280} {\bibfield  {journal} {\bibinfo
  {journal} {Zeitschrift f{\"u}r Physik}\ }\textbf {\bibinfo {volume} {43}},\
  \bibinfo {pages} {172} (\bibinfo {year} {1927})}\BibitemShut {NoStop}%
\bibitem [{\citenamefont {Coles}\ \emph {et~al.}(2017)\citenamefont {Coles},
  \citenamefont {Berta}, \citenamefont {Tomamichel},\ and\ \citenamefont
  {Wehner}}]{ColesPRMP17}%
  \BibitemOpen
  \bibfield  {author} {\bibinfo {author} {\bibfnamefont {P.~J.}\ \bibnamefont
  {Coles}}, \bibinfo {author} {\bibfnamefont {M.}~\bibnamefont {Berta}},
  \bibinfo {author} {\bibfnamefont {M.}~\bibnamefont {Tomamichel}},\ and\
  \bibinfo {author} {\bibfnamefont {S.}~\bibnamefont {Wehner}},\ }\bibfield
  {title} {\bibinfo {title} {Entropic uncertainty relations and their
  applications},\ }\href {https://doi.org/10.1103/RevModPhys.89.015002}
  {\bibfield  {journal} {\bibinfo  {journal} {Rev. Mod. Phys.}\ }\textbf
  {\bibinfo {volume} {89}},\ \bibinfo {pages} {015002} (\bibinfo {year}
  {2017})}\BibitemShut {NoStop}%
\bibitem [{\citenamefont {Cerf}\ and\ \citenamefont {Adami}(1997)}]{Cerf97}%
  \BibitemOpen
  \bibfield  {author} {\bibinfo {author} {\bibfnamefont {N.~J.}\ \bibnamefont
  {Cerf}}\ and\ \bibinfo {author} {\bibfnamefont {C.}~\bibnamefont {Adami}},\
  }\bibfield  {title} {\bibinfo {title} {Negative entropy and information in
  quantum mechanics},\ }\href {https://doi.org/10.1103/PhysRevLett.79.5194}
  {\bibfield  {journal} {\bibinfo  {journal} {Phys. Rev. Lett.}\ }\textbf
  {\bibinfo {volume} {79}},\ \bibinfo {pages} {5194} (\bibinfo {year}
  {1997})}\BibitemShut {NoStop}%
\bibitem [{\citenamefont {Konig}\ \emph {et~al.}(2009)\citenamefont {Konig},
  \citenamefont {Renner},\ and\ \citenamefont
  {Schaffner}}]{Konig2009operational}%
  \BibitemOpen
  \bibfield  {author} {\bibinfo {author} {\bibfnamefont {R.}~\bibnamefont
  {Konig}}, \bibinfo {author} {\bibfnamefont {R.}~\bibnamefont {Renner}},\ and\
  \bibinfo {author} {\bibfnamefont {C.}~\bibnamefont {Schaffner}},\ }\bibfield
  {title} {\bibinfo {title} {The operational meaning of min-and max-entropy},\
  }\href {https://doi.org/10.1109/TIT.2009.2025545} {\bibfield  {journal}
  {\bibinfo  {journal} {IEEE Transactions on Information theory}\ }\textbf
  {\bibinfo {volume} {55}},\ \bibinfo {pages} {4337} (\bibinfo {year}
  {2009})}\BibitemShut {NoStop}%
\bibitem [{\citenamefont {Horodecki}\ \emph {et~al.}(2005)\citenamefont
  {Horodecki}, \citenamefont {Oppenheim},\ and\ \citenamefont
  {Winter}}]{Horodecki2005partial}%
  \BibitemOpen
  \bibfield  {author} {\bibinfo {author} {\bibfnamefont {M.}~\bibnamefont
  {Horodecki}}, \bibinfo {author} {\bibfnamefont {J.}~\bibnamefont
  {Oppenheim}},\ and\ \bibinfo {author} {\bibfnamefont {A.}~\bibnamefont
  {Winter}},\ }\bibfield  {title} {\bibinfo {title} {Partial quantum
  information},\ }\href {https://doi.org/10.1038/nature03909} {\bibfield
  {journal} {\bibinfo  {journal} {Nature}\ }\textbf {\bibinfo {volume} {436}},\
  \bibinfo {pages} {673} (\bibinfo {year} {2005})}\BibitemShut {NoStop}%
\bibitem [{\citenamefont {Rio}\ \emph {et~al.}(2011)\citenamefont {Rio},
  \citenamefont {{\AA}berg}, \citenamefont {Renner}, \citenamefont {Dahlsten},\
  and\ \citenamefont {Vedral}}]{Rio2011thermodynamic}%
  \BibitemOpen
  \bibfield  {author} {\bibinfo {author} {\bibfnamefont {L.~d.}\ \bibnamefont
  {Rio}}, \bibinfo {author} {\bibfnamefont {J.}~\bibnamefont {{\AA}berg}},
  \bibinfo {author} {\bibfnamefont {R.}~\bibnamefont {Renner}}, \bibinfo
  {author} {\bibfnamefont {O.}~\bibnamefont {Dahlsten}},\ and\ \bibinfo
  {author} {\bibfnamefont {V.}~\bibnamefont {Vedral}},\ }\bibfield  {title}
  {\bibinfo {title} {The thermodynamic meaning of negative entropy},\ }\href
  {https://doi.org/10.1038/nature10123} {\bibfield  {journal} {\bibinfo
  {journal} {Nature}\ }\textbf {\bibinfo {volume} {474}},\ \bibinfo {pages}
  {61} (\bibinfo {year} {2011})}\BibitemShut {NoStop}%
\bibitem [{\citenamefont {Brown}\ \emph {et~al.}(2021)\citenamefont {Brown},
  \citenamefont {Fawzi},\ and\ \citenamefont {Fawzi}}]{Brown2021computing}%
  \BibitemOpen
  \bibfield  {author} {\bibinfo {author} {\bibfnamefont {P.}~\bibnamefont
  {Brown}}, \bibinfo {author} {\bibfnamefont {H.}~\bibnamefont {Fawzi}},\ and\
  \bibinfo {author} {\bibfnamefont {O.}~\bibnamefont {Fawzi}},\ }\bibfield
  {title} {\bibinfo {title} {Computing conditional entropies for quantum
  correlations},\ }\href {https://doi.org/10.1038/s41467-020-20018-1}
  {\bibfield  {journal} {\bibinfo  {journal} {Nature communications}\ }\textbf
  {\bibinfo {volume} {12}},\ \bibinfo {pages} {575} (\bibinfo {year}
  {2021})}\BibitemShut {NoStop}%
\bibitem [{\citenamefont {Tomamichel}(2015)}]{Tomamichel2015quantum1}%
  \BibitemOpen
  \bibfield  {author} {\bibinfo {author} {\bibfnamefont {M.}~\bibnamefont
  {Tomamichel}},\ }\href {https://doi.org/10.1007/978-3-319-21891-5} {\bibinfo
  {title} {Quantum information processing with finite resources: mathematical
  foundations}} (\bibinfo {year} {2015})\BibitemShut {NoStop}%
\bibitem [{\citenamefont {Bennett}\ \emph {et~al.}(1996)\citenamefont
  {Bennett}, \citenamefont {DiVincenzo}, \citenamefont {Smolin},\ and\
  \citenamefont {Wootters}}]{BennettPRA96}%
  \BibitemOpen
  \bibfield  {author} {\bibinfo {author} {\bibfnamefont {C.~H.}\ \bibnamefont
  {Bennett}}, \bibinfo {author} {\bibfnamefont {D.~P.}\ \bibnamefont
  {DiVincenzo}}, \bibinfo {author} {\bibfnamefont {J.~A.}\ \bibnamefont
  {Smolin}},\ and\ \bibinfo {author} {\bibfnamefont {W.~K.}\ \bibnamefont
  {Wootters}},\ }\bibfield  {title} {\bibinfo {title} {Mixed-state entanglement
  and quantum error correction},\ }\href
  {https://doi.org/10.1103/PhysRevA.54.3824} {\bibfield  {journal} {\bibinfo
  {journal} {Phys. Rev. A}\ }\textbf {\bibinfo {volume} {54}},\ \bibinfo
  {pages} {3824} (\bibinfo {year} {1996})}\BibitemShut {NoStop}%
\bibitem [{\citenamefont {\ifmmode~\dot{Z}\else \.{Z}\fi{}yczkowski}\ \emph
  {et~al.}(1998)\citenamefont {\ifmmode~\dot{Z}\else \.{Z}\fi{}yczkowski},
  \citenamefont {Horodecki}, \citenamefont {Sanpera},\ and\ \citenamefont
  {Lewenstein}}]{Zyczkowski98}%
  \BibitemOpen
  \bibfield  {author} {\bibinfo {author} {\bibfnamefont {K.}~\bibnamefont
  {\ifmmode~\dot{Z}\else \.{Z}\fi{}yczkowski}}, \bibinfo {author}
  {\bibfnamefont {P.}~\bibnamefont {Horodecki}}, \bibinfo {author}
  {\bibfnamefont {A.}~\bibnamefont {Sanpera}},\ and\ \bibinfo {author}
  {\bibfnamefont {M.}~\bibnamefont {Lewenstein}},\ }\bibfield  {title}
  {\bibinfo {title} {Volume of the set of separable states},\ }\href
  {https://doi.org/10.1103/PhysRevA.58.883} {\bibfield  {journal} {\bibinfo
  {journal} {Phys. Rev. A}\ }\textbf {\bibinfo {volume} {58}},\ \bibinfo
  {pages} {883} (\bibinfo {year} {1998})}\BibitemShut {NoStop}%
\bibitem [{\citenamefont {Vidal}\ and\ \citenamefont {Werner}(2002)}]{Vidal02}%
  \BibitemOpen
  \bibfield  {author} {\bibinfo {author} {\bibfnamefont {G.}~\bibnamefont
  {Vidal}}\ and\ \bibinfo {author} {\bibfnamefont {R.~F.}\ \bibnamefont
  {Werner}},\ }\bibfield  {title} {\bibinfo {title} {Computable measure of
  entanglement},\ }\href {https://doi.org/10.1103/PhysRevA.65.032314}
  {\bibfield  {journal} {\bibinfo  {journal} {Phys. Rev. A}\ }\textbf {\bibinfo
  {volume} {65}},\ \bibinfo {pages} {032314} (\bibinfo {year}
  {2002})}\BibitemShut {NoStop}%
\bibitem [{\citenamefont {Christandl}\ and\ \citenamefont
  {Winter}(2004)}]{Christandl2004squashed}%
  \BibitemOpen
  \bibfield  {author} {\bibinfo {author} {\bibfnamefont {M.}~\bibnamefont
  {Christandl}}\ and\ \bibinfo {author} {\bibfnamefont {A.}~\bibnamefont
  {Winter}},\ }\bibfield  {title} {\bibinfo {title} {“squashed
  entanglement”: an additive entanglement measure},\ }\href
  {https://doi.org/10.1063/1.1643788} {\bibfield  {journal} {\bibinfo
  {journal} {Journal of mathematical physics}\ }\textbf {\bibinfo {volume}
  {45}},\ \bibinfo {pages} {829} (\bibinfo {year} {2004})}\BibitemShut
  {NoStop}%
\bibitem [{\citenamefont {Bu}\ \emph {et~al.}(2023{\natexlab{c}})\citenamefont
  {Bu}, \citenamefont {Gu},\ and\ \citenamefont {Jaffe}}]{BGJ23c}%
  \BibitemOpen
  \bibfield  {author} {\bibinfo {author} {\bibfnamefont {K.}~\bibnamefont
  {Bu}}, \bibinfo {author} {\bibfnamefont {W.}~\bibnamefont {Gu}},\ and\
  \bibinfo {author} {\bibfnamefont {A.}~\bibnamefont {Jaffe}},\ }\bibfield
  {title} {\bibinfo {title} {Stabilizer testing and magic entropy},\ }\href
  {https://arxiv.org/abs/2306.09292} {\bibfield  {journal} {\bibinfo  {journal}
  {arXiv:2306.09292}\ } (\bibinfo {year} {2023}{\natexlab{c}})}\BibitemShut
  {NoStop}%
\bibitem [{\citenamefont {Bu}\ and\ \citenamefont {Jaffe}(2024)}]{BJ24a}%
  \BibitemOpen
  \bibfield  {author} {\bibinfo {author} {\bibfnamefont {K.}~\bibnamefont
  {Bu}}\ and\ \bibinfo {author} {\bibfnamefont {A.}~\bibnamefont {Jaffe}},\
  }\bibfield  {title} {\bibinfo {title} {Magic can enhance the quantum capacity
  of channels},\ }\href {https://arxiv.org/abs/2401.12105} {\bibfield
  {journal} {\bibinfo  {journal} {arXiv:2401.12105}\ } (\bibinfo {year}
  {2024})}\BibitemShut {NoStop}%
\bibitem [{\citenamefont {Bu}\ \emph {et~al.}(2024{\natexlab{c}})\citenamefont
  {Bu}, \citenamefont {Jaffe},\ and\ \citenamefont {Wei}}]{BJW24a}%
  \BibitemOpen
  \bibfield  {author} {\bibinfo {author} {\bibfnamefont {K.}~\bibnamefont
  {Bu}}, \bibinfo {author} {\bibfnamefont {A.}~\bibnamefont {Jaffe}},\ and\
  \bibinfo {author} {\bibfnamefont {Z.}~\bibnamefont {Wei}},\ }\bibfield
  {title} {\bibinfo {title} {Magic class and the convolution group},\ }\href
  {https://arxiv.org/abs/2402.05780} {\bibfield  {journal} {\bibinfo  {journal}
  {arXiv:2402.05780}\ } (\bibinfo {year} {2024}{\natexlab{c}})}\BibitemShut
  {NoStop}%
\bibitem [{\citenamefont
  {Zamolodchikov}(1986)}]{Zamolodchikov1986irreversibility}%
  \BibitemOpen
  \bibfield  {author} {\bibinfo {author} {\bibfnamefont {A.~B.}\ \bibnamefont
  {Zamolodchikov}},\ }\bibfield  {title} {\bibinfo {title} {Irreversibility of
  the flux of the renormalization group in a 2d field theory},\ }\href@noop {}
  {\bibfield  {journal} {\bibinfo  {journal} {JETP lett}\ }\textbf {\bibinfo
  {volume} {43}},\ \bibinfo {pages} {730} (\bibinfo {year} {1986})}\BibitemShut
  {NoStop}%
\bibitem [{\citenamefont {Cardy}(1988)}]{Cardy1988there}%
  \BibitemOpen
  \bibfield  {author} {\bibinfo {author} {\bibfnamefont {J.~L.}\ \bibnamefont
  {Cardy}},\ }\bibfield  {title} {\bibinfo {title} {Is there a c-theorem in
  four dimensions?},\ }\href {https://doi.org/10.1016/0370-2693(88)90054-8}
  {\bibfield  {journal} {\bibinfo  {journal} {Physics Letters B}\ }\textbf
  {\bibinfo {volume} {215}},\ \bibinfo {pages} {749} (\bibinfo {year}
  {1988})}\BibitemShut {NoStop}%
\bibitem [{\citenamefont {Osborn}(1989)}]{Osborn1989derivation}%
  \BibitemOpen
  \bibfield  {author} {\bibinfo {author} {\bibfnamefont {H.}~\bibnamefont
  {Osborn}},\ }\bibfield  {title} {\bibinfo {title} {Derivation of a four
  dimensional c-theorem for renormaliseable quantum field theories},\ }\href
  {https://doi.org/10.1016/0370-2693(89)90729-6} {\bibfield  {journal}
  {\bibinfo  {journal} {Physics Letters B}\ }\textbf {\bibinfo {volume}
  {222}},\ \bibinfo {pages} {97} (\bibinfo {year} {1989})}\BibitemShut
  {NoStop}%
\bibitem [{\citenamefont {Komargodski}\ and\ \citenamefont
  {Schwimmer}(2011)}]{Komargodski2011renormalization}%
  \BibitemOpen
  \bibfield  {author} {\bibinfo {author} {\bibfnamefont {Z.}~\bibnamefont
  {Komargodski}}\ and\ \bibinfo {author} {\bibfnamefont {A.}~\bibnamefont
  {Schwimmer}},\ }\bibfield  {title} {\bibinfo {title} {On renormalization
  group flows in four dimensions},\ }\href
  {https://doi.org/10.1007/JHEP12(2011)099} {\bibfield  {journal} {\bibinfo
  {journal} {Journal of High Energy Physics}\ }\textbf {\bibinfo {volume}
  {2011}},\ \bibinfo {pages} {1} (\bibinfo {year} {2011})}\BibitemShut
  {NoStop}%
\bibitem [{\citenamefont {Nishioka}(2018)}]{NishiokaRMP18}%
  \BibitemOpen
  \bibfield  {author} {\bibinfo {author} {\bibfnamefont {T.}~\bibnamefont
  {Nishioka}},\ }\bibfield  {title} {\bibinfo {title} {Entanglement entropy:
  Holography and renormalization group},\ }\href
  {https://doi.org/10.1103/RevModPhys.90.035007} {\bibfield  {journal}
  {\bibinfo  {journal} {Rev. Mod. Phys.}\ }\textbf {\bibinfo {volume} {90}},\
  \bibinfo {pages} {035007} (\bibinfo {year} {2018})}\BibitemShut {NoStop}%
\bibitem [{\citenamefont {Berta}\ \emph {et~al.}(2010)\citenamefont {Berta},
  \citenamefont {Christandl}, \citenamefont {Colbeck}, \citenamefont {Renes},\
  and\ \citenamefont {Renner}}]{Berta2010uncertainty}%
  \BibitemOpen
  \bibfield  {author} {\bibinfo {author} {\bibfnamefont {M.}~\bibnamefont
  {Berta}}, \bibinfo {author} {\bibfnamefont {M.}~\bibnamefont {Christandl}},
  \bibinfo {author} {\bibfnamefont {R.}~\bibnamefont {Colbeck}}, \bibinfo
  {author} {\bibfnamefont {J.~M.}\ \bibnamefont {Renes}},\ and\ \bibinfo
  {author} {\bibfnamefont {R.}~\bibnamefont {Renner}},\ }\bibfield  {title}
  {\bibinfo {title} {The uncertainty principle in the presence of quantum
  memory},\ }\href {https://doi.org/10.1038/nphys1734} {\bibfield  {journal}
  {\bibinfo  {journal} {Nature Physics}\ }\textbf {\bibinfo {volume} {6}},\
  \bibinfo {pages} {659} (\bibinfo {year} {2010})}\BibitemShut {NoStop}%
\bibitem [{\citenamefont {Choi}(1975)}]{Choi75}%
  \BibitemOpen
  \bibfield  {author} {\bibinfo {author} {\bibfnamefont {M.-D.}\ \bibnamefont
  {Choi}},\ }\bibfield  {title} {\bibinfo {title} {Completely positive linear
  maps on complex matrices},\ }\href
  {https://doi.org/10.1016/0024-3795(75)90075-0} {\bibfield  {journal}
  {\bibinfo  {journal} {Linear Algebra and its Application}\ }\textbf {\bibinfo
  {volume} {10}},\ \bibinfo {pages} {285} (\bibinfo {year} {1975})}\BibitemShut
  {NoStop}%
\bibitem [{\citenamefont {Jamiołkowski}(1972)}]{Jamio72}%
  \BibitemOpen
  \bibfield  {author} {\bibinfo {author} {\bibfnamefont {A.}~\bibnamefont
  {Jamiołkowski}},\ }\bibfield  {title} {\bibinfo {title} {Linear
  transformations which preserve trace and positive semidefiniteness of
  operators},\ }\href {https://doi.org/10.1016/0034-4877(72)90011-0} {\bibfield
   {journal} {\bibinfo  {journal} {Rep. Math. Phys.}\ }\textbf {\bibinfo
  {volume} {3}},\ \bibinfo {pages} {275–278} (\bibinfo {year}
  {1972})}\BibitemShut {NoStop}%
\bibitem [{\citenamefont {Tao}(2003)}]{Tao2003uncertainty}%
  \BibitemOpen
  \bibfield  {author} {\bibinfo {author} {\bibfnamefont {T.}~\bibnamefont
  {Tao}},\ }\bibfield  {title} {\bibinfo {title} {An uncertainty principle for
  cyclic groups of prime order},\ }\href@noop {} {\bibfield  {journal}
  {\bibinfo  {journal} {arXiv preprint math/0308286}\ } (\bibinfo {year}
  {2003})}\BibitemShut {NoStop}%
\bibitem [{\citenamefont {Donoho}\ and\ \citenamefont
  {Stark}(1989)}]{Donoho1989uncertainty}%
  \BibitemOpen
  \bibfield  {author} {\bibinfo {author} {\bibfnamefont {D.~L.}\ \bibnamefont
  {Donoho}}\ and\ \bibinfo {author} {\bibfnamefont {P.~B.}\ \bibnamefont
  {Stark}},\ }\bibfield  {title} {\bibinfo {title} {Uncertainty principles and
  signal recovery},\ }\href {https://doi.org/10.1137/014905} {\bibfield
  {journal} {\bibinfo  {journal} {SIAM Journal on Applied Mathematics}\
  }\textbf {\bibinfo {volume} {49}},\ \bibinfo {pages} {906} (\bibinfo {year}
  {1989})}\BibitemShut {NoStop}%
\bibitem [{\citenamefont {Smith}(1990)}]{Smith1990uncertainty}%
  \BibitemOpen
  \bibfield  {author} {\bibinfo {author} {\bibfnamefont {K.~T.}\ \bibnamefont
  {Smith}},\ }\bibfield  {title} {\bibinfo {title} {The uncertainty principle
  on groups},\ }\href {https://doi.org/10.1137/0150051} {\bibfield  {journal}
  {\bibinfo  {journal} {SIAM Journal on Applied Mathematics}\ }\textbf
  {\bibinfo {volume} {50}},\ \bibinfo {pages} {876} (\bibinfo {year}
  {1990})}\BibitemShut {NoStop}%
\bibitem [{\citenamefont {Jiang}\ \emph {et~al.}(2016)\citenamefont {Jiang},
  \citenamefont {Liu},\ and\ \citenamefont {Wu}}]{Jiang2016noncommutative}%
  \BibitemOpen
  \bibfield  {author} {\bibinfo {author} {\bibfnamefont {C.}~\bibnamefont
  {Jiang}}, \bibinfo {author} {\bibfnamefont {Z.}~\bibnamefont {Liu}},\ and\
  \bibinfo {author} {\bibfnamefont {J.}~\bibnamefont {Wu}},\ }\bibfield
  {title} {\bibinfo {title} {Noncommutative uncertainty principles},\ }\href
  {https://doi.org/j.jfa.2015.08.007} {\bibfield  {journal} {\bibinfo
  {journal} {Journal of Functional Analysis}\ }\textbf {\bibinfo {volume}
  {270}},\ \bibinfo {pages} {264} (\bibinfo {year} {2016})}\BibitemShut
  {NoStop}%
\end{thebibliography}%
\clearpage
\newpage
\onecolumngrid
\appendix
\section{Extremality of stabilizer states in uncertainty principles}\label{appen:prop_con}
In this section, we establish three uncertainty principles corresponding to  three different,
Fourier transformations, as illustrated in the Figure \ref{fig:FT}, where the 
stabilizer states are the only extremizers.  

\begin{figure}
    \includegraphics[width=8.0cm]{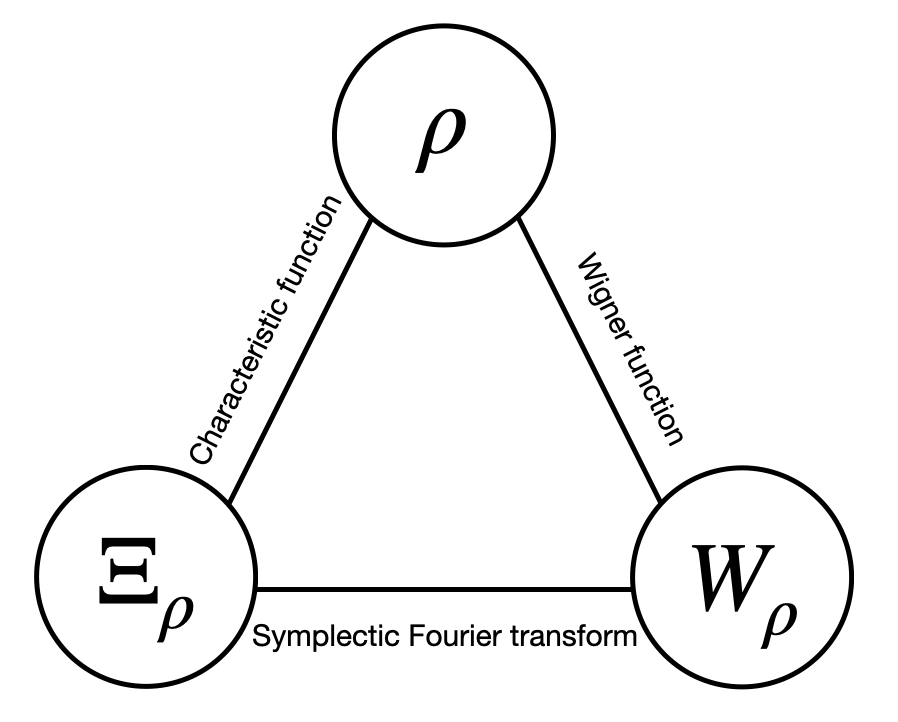}
    \caption{We consider three different Fourier transforms: (1) using the characteristic function as the Fourier coefficients (2) using the Wigner function 
    as the Fourier coefficients (3) applying the symplectic Fourier transform which will map the characteristic function to the Wigner function.}
    \label{fig:FT}
\end{figure}

\begin{thm}[Restatement of the uncertainty principle for Pauli rank]
    Given an $n$-qudit state $\rho$, we have 
    \begin{eqnarray}
S_{\max}(\rho)+\log \chi_P(\rho)
        \geq n\log d,
    \end{eqnarray}
    and the equality holds iff $\rho$ is a stabilizer state, i.e.,  $\rho=\mathcal{M}(\rho)$.
\end{thm}
\begin{proof}
    The proof of the inequality is similar to the method in the work~\cite{Tao2003uncertainty,Donoho1989uncertainty,Smith1990uncertainty,Jiang2016noncommutative}. Here, by
    the Cauchy-Schwarz inequality, we have
\begin{eqnarray}
  1\leq \left(\Tr{\rho^2}\right)^{1/2}\sqrt{\text{Rank}(\rho)}
=\left(\frac{1}{d^n}\sum_{\vec x\in V^n}|\Xi_{\rho}(\vec x)|^2\right)^{1/2}\sqrt{\text{Rank}(\rho)}
\leq \left(\frac{1}{d^n}\chi_P(\rho)\right)^{1/2}\sqrt{\text{Rank}(\rho)},
\end{eqnarray}
where the equality comes from the Parseval's identity $\frac{1}{d^n}\sum_{\vec x\in V^n}|\Xi_{\rho}(\vec x)|^2=\Tr{\rho^2}$,  and the last
inequality comes from the fact that $|\Xi_{\rho}(\vec x)|\leq 1$.

Hence, the equality implies that the $|\Xi_{\rho}(\vec x)|=1$ for any $\vec x\in V^n$ with $\Xi_{\rho}(\vec x)\neq 0$.
Hence $\rho=\mathcal{M}(\rho)$ by the definition of $\mathcal{M}(\rho)$. Conversely, 
if $\rho=\mathcal{M}(\rho)$, it is straightforward to verify the condition for equality.

\end{proof}

\begin{thm}[Restatement of the uncertainty principle for Wigner rank]
Given an $n$-qudit state $\rho$ with odd prime $d$, we have 
    \begin{eqnarray}
      S_{\max}(\rho)+\log \chi_W(\rho)
        \geq n\log d,
    \end{eqnarray}
    and the equality holds iff $\rho$ is a pure stabilizer state.
\end{thm}
\begin{proof}
    The proof of the inequality is similar to the above theorem. 
     By
    the Cauchy-Schwarz inequality, we have
    \begin{align}
        1=\sum_{\vec x}W_{\rho}(\vec x)\leq \sum_{\vec x}|W_{\rho}(\vec x)|
    \leq \left(\sum_{\vec x}W_{\rho}(\vec x)^2\right)^{1/2}\sqrt{\chi_W(\rho)}
= \left(\frac{1}{d^n}\Tr{\rho^2}\right)^{1/2}\sqrt{\chi_W(\rho)}
    \leq \left(\frac{1}{d^n}\text{Rank}(\rho)\right)^{1/2}\sqrt{\chi_W(\rho)},
    \end{align}
where the equality comes from the Parseval's identity
$\Tr{\rho^2}=d^n\sum_{\vec x\in V^n}W_{\rho}(\vec x)^2$,
and the last
inequality comes from the fact that $\Tr{\rho^2}\leq \text{Rank}(\rho)$.

Hence, the equality implies that the $\Tr{\rho^2}= \text{Rank}(\rho)$, i.e., $\rho$ is a pure state, and 
$W_{\rho}$ is constant on its support, which is equal to $\frac{1}{d^n}$. Then, by the discrete Hudson theorem \cite{Gross06}, 
$\rho$ is a pure stabilizer state. 
Conversely, 
if $\rho$ is a pure stabilizer state, it is straightforward to verify the condition for equality.

\end{proof}

\begin{prop}[Restatement of Proposition~\ref{prop:xw}]
    Given an $n$-qudit state $\rho$ with odd prime $d$, we have 
    \begin{eqnarray}
        \log \chi_P(\rho)+\log \chi_W(\rho)
        \geq 2n\log d,
    \end{eqnarray}
    and the equality holds iff $\rho$ is a stabilizer state, i.e.,  $\rho=\mathcal{M}(\rho)$.
\end{prop}
\begin{proof}
The proof of the inequality is similar to the above theorem.   By the Cauchy-Schwarz inequality, we have 
\begin{align}
    1=\sum_{\vec x}W_{\rho}(\vec x)
    \leq \left(\sum_{\vec x}|W_{\rho}(\vec x)|^2\right)^{1/2}\sqrt{\chi_W(\rho)}
    = \left(\frac{1}{d^{2n}}\sum_{\vec x}|\Xi_{\rho}(\vec x)|^2\right)^{1/2}\sqrt{\chi_W(\rho)}
    \leq \left(\frac{1}{d^{2n}} \chi_P(\rho)\right)^{1/2}\sqrt{\chi_W(\rho)},
\end{align}
where the equality comes from the Parseval's identity
$\Tr{\rho^2}=d^n\sum_{\vec x\in V^n}W_{\rho}(\vec x)^2=\frac{1}{d^n}\sum_{\vec x\in V^n}|\Xi_{\rho}(\vec x)|^2$,
and the last
inequality comes from the fact that $|\Xi_{\rho}(\vec x)|\leq 1$.

Hence, the equality holds implies that the $|\Xi_{\rho}(\vec x)|=1$ for any $\vec x\in V^n$ such that $\Xi_{\rho}(\vec x)\neq 0$.
Hence $\rho=\mathcal{M}(\rho)$ by the definition of $\mathcal{M}(\rho)$. Conversely, 
if $\rho=\mathcal{M}(\rho)$, it is straightforward to verify the condition for equality.

\end{proof}

\section{Extremality of stabilizer states in the information and correlation measures}
\begin{thm}[Restatement of the general result]
     Let $F:D(\otimes^m_i\mathcal{H}_i)\to\real$  be a convex function, 
    where $\otimes^m_i\mathcal{H}_i$ is an $m$-partite system and each subsystem $\mathcal{H}_i$ consists of $n_i$ qudits. If $F$ is
    invariant under local unitary,
    we have 
    \begin{eqnarray}
        F(\rho)\geq F(\mathcal{M}(\rho)),
    \end{eqnarray}
    where $\mathcal{M}(\rho)$ is the stabilizer  state with the same stabilizer group as $\rho$.
\end{thm}

\begin{proof}
Although $\mathcal{M}$ is not a quantum channel, $\mathcal{M}(\rho)$
can be written as a convex combination of $w(\vec x)\rho w(\vec x)^\dag$, where $w(\vec x)$ depends on $\rho$. 
Let us consider the stabilizer group $G_{\rho}$ of the state $\rho$, and assume that the generators 
of $G_{\rho}$ are $\set{w(\vec x_i)}_{i\in [r]}$. 
Hence, there exists some Clifford unitary such that $Uw(\vec x_i)U^\dag=Z_i$, where
$Z_i$ is the Pauli $Z$ operator on $i$-th qudit. 
Hence, the stabilizer group of $U\rho U^\dag$ is generated by 
$\set{Z_i}_{i\in[r]}$, which leads to
\begin{eqnarray}
    \mathcal{M}(U\rho U^\dag)=\frac{1}{d^{2(n-r)}}\sum_{\vec y\in V^{n-r}}
    w(\vec 0, \vec y)U\rho U^\dag
    w(\vec 0, \vec y)^\dag.
    \end{eqnarray}
Since $\mathcal{M}$ is commuting with Clifford unitary (see Lemma 15 in \cite{BGJ23b}), 
we have 
\begin{eqnarray}
    \mathcal{M}(\rho)
    =\frac{1}{d^{2(n-r)}}\sum_{\vec y\in V^{n-r}}
    U^\dag w(\vec 0, \vec y)U\rho U^\dag
    w(\vec 0, \vec y)^\dag U
    =\frac{1}{d^{2(n-r)}}\sum_{\vec y\in V^{n-r}}
    w(\vec y')\rho w(\vec y')^\dag,
\end{eqnarray}
where $ w(\vec y')=U^\dag w(\vec 0, \vec y)U$  is a Pauli operator as Clifford unitaries will map operators to Pauli operators. 
Therefore,  
\begin{eqnarray}
    F(\mathcal{M}(\rho))
    =F\left(\frac{1}{d^{2(n-r)}}\sum_{\vec y\in V^{n-r}}
    w(\vec y')\rho w(\vec y')^\dag\right)
    \leq \frac{1}{d^{2(n-r)}}\sum_{\vec y\in V^{n-r}}F( w(\vec y')\rho w(\vec y')^\dag)
    =F(\rho),
\end{eqnarray}
where the first inequality comes from the assumption that $F$ is convex, and the last equality comes
from the assumption that $F$ is invariant under local unitaries.
    
\end{proof}

\begin{prop}[Restatement of the monotonicity of entanglement entropy]
 The entanglement entropy is monotonically increasing under quantum convolution, i.e., 
 \begin{eqnarray}
     S(A)_{\boxtimes_L\rho}
     \leq   S(A)_{\boxtimes_{L+1}\rho},
 \end{eqnarray}
 for any $L\geq 0$, where $S(A)_{\boxtimes_L\rho}=S((\boxtimes_L\rho)_A)$ is the entanglement entropy of subsystem $A$ after applying the $L$-th iteration of the quantum convolution to the state $\rho$.
\end{prop}
\begin{proof}
    Because $\Ptr{A^c}{\boxtimes_L\rho}=\boxtimes_L\Ptr{A^c}{\rho}$ where $A^c$ is the complement of $A$, we have 
    $S(A)_{\boxtimes_L\rho}=S(\boxtimes_L\rho_A)$. Then we get the result by the monotonicity of quantum entropy under convolution.
\end{proof}

\begin{prop}[Restatement of the monotonicity of conditional entropy]
 The conditional entropy $S_{\alpha}(A|B)$ is monotonically increasing under convolution for any $\alpha\geq 1/2$, i.e., 
 \begin{eqnarray}
     S_{\alpha}(A|B)_{\boxtimes_L\rho}
     \leq   S_{\alpha}(A|B)_{\boxtimes_{L+1}\rho},
 \end{eqnarray}
 for any $L\geq 0$.
\end{prop}
\begin{proof}
First, by the data processing inequality for $\alpha\geq 1/2$, we have
\begin{eqnarray}
    D_{\alpha}\left(\rho_{AB}||\frac{I_A}{d_A}\ot \rho_B\right)
    \geq  D_{\alpha}\left(\rho_{AB}\boxtimes\rho_{AB}||\left(\frac{I_A}{d_A}\ot \rho_B\right)\boxtimes\rho_{AB}\right)
    = D_{\alpha}\left(\rho_{AB}\boxtimes\rho_{AB}||\frac{I_A}{d_A}\ot (\rho_B\boxtimes\rho_B)\right).
\end{eqnarray}
Hence, by the definition  of $S_{\alpha}(A|B)$, we have $S_{\alpha}(A|B)_{\rho}\leq S(A|B)_{\rho\boxtimes\rho}$. 
Repeat the above process, we can get $ S_{\alpha}(A|B)_{\boxtimes_L\rho}
     \leq   S_{\alpha}(A|B)_{\boxtimes_{L+1}\rho}$ for any $L$.
\end{proof}

\end{document}